\documentclass[journal,10pt,onecolumn,draftclsnofoot]{IEEEtran}  

\IEEEoverridecommandlockouts                              

\usepackage{amsmath, amsfonts, dsfont, amssymb, mathrsfs, setspace, graphicx, epsfig, amsthm,url,float,framed,circuitikz, color,subcaption,epstopdf}
\usepackage{multicol}

\newtheorem{theorem}{Theorem}

\newtheorem{definition}{Definition}
\newtheorem{lemma}{Lemma}
\newtheorem{corollary}{Corollary}

\newtheorem{prop}{Proposition}
\newtheorem{remark}{Remark}

\newtheorem{assumption}{Assumption}

\makeatletter
\usepackage[noadjust]{cite}
\makeatother

\def\scr#1{{\cal #1}}
\newcommand{\R}{\mathbb{R}}
\newcommand{\Z}{\mathbb{Z}}
\newcommand{\D}{\mathcal{D}}
\newcommand{\0}{\mathbf{0}}
\newcommand{\1}{\mathbf{1}}

\title{
Analysis, 
Online Estimation, and Validation of a Competing Virus Model}

\author{Philip E. Par\'{e}, Damir Vrabac, Henrik Sandberg, and Karl H. Johansson*\thanks{
* 
The authors are all with the Division of Decision and Control Systems at KTH Royal Institute of Technology  
and Philip E. Par\'{e} can be contacted at  (\texttt{philipar@kth.se}).
}}
\begin{document}
\maketitle

\begin{abstract}
In this paper we introduce a discrete time competing virus model and the assumptions necessary for the model to be well posed. 
We analyze the system exploring its different equilibria. 
We provide necessary and sufficient conditions for the estimation of the model parameters from time series data  
and introduce an online estimation algorithm. 
We employ a dataset of two competing subsidy programs from the US Department of Agriculture to validate the model by employing the identification techniques. 
To the best of our knowledge, this work is the first to study competing virus models in discrete-time, online identification of spread parameters from time series data, and validation of said models using real data. These new contributions are important for  applications since real data is naturally sampled.
\end{abstract}

\section{Introduction}

As the world becomes more connected via transportation networks, communication networks, social media, and others, society  become more susceptible to various types of attacks such as diseases, viruses, and misinformation (fake news). 
We have witnessed the massive impacts that the spread of misinformation can have, especially in political systems \cite{allcott2017social,bastos2019brexit}. Therefore, it is important to develop models that capture the behavior of spreading competing information to be able to design and implement mitigation techniques against fake news.

Competing virus models have been motivated in the literature by competing viral strains \cite{nowak1991evolution} and competing ideas spreading on different social networks \cite{sahneh2014competitive}, but they can also have broader applications to political stances, adoption of competing products, competing practices in farming, etc. 
Competing SIS virus models have been studied extensively in recent years \cite{nowak1991evolution,karrer2011competing,prakash2012winner,wei2013competing,sahneh2014competitive,santos2015bivirus,liu2016onthe,watkins2016optimal,bivirusTAC,xu2012multi,acc_multi,tnse_multi}. 
In \cite{nowak1991evolution}, the idea of modeling two competing viruses was introduced without any graph structure. The more recent works have included graph structure. The majority of this work has focused on the case of two competing viruses, sometimes referred to as the bi-virus model \cite{karrer2011competing,prakash2012winner,wei2013competing,sahneh2014competitive,santos2015bivirus,liu2016onthe,watkins2016optimal,bivirusTAC}. 
Some work has 
analyzed the equilibria of models of 
an arbitrary number of competing viruses \cite{xu2012multi,acc_multi,tnse_multi}. 
To the best of our knowledge all of the previous work on competing viruses has been done in continuous time. 

Discrete time models have been studied for the single virus model \cite{wang2003epidemic,chakrabarti2008epidemic,ahn2013global,usda_acc,dtjournal,arturo,prasse2018network,prasse2019viral}. 
In \cite{usda_acc,dtjournal,arturo,prasse2018network} identification of a single virus discrete time spread processes was investigated. In \cite{prasse2018network}, in addition to recovering the homogeneous spread parameters the authors studied recovering the network structure of the model, but no real data was employed. In \cite{usda_acc,dtjournal,arturo}, validation work was carried out using real data. One dataset used in \cite{usda_acc,dtjournal} was the adoption of two competing US Department of Agriculture (USDA) farm subsidy programs. We employ that dataset here but use a two-competing virus (bi-virus) model. The results show that the model fits the dataset much better than when using a single virus model. 

In many ways this paper is an extension of \cite{usda_acc,dtjournal}, generalizing from a single virus discrete model to multiple competing viruses. However, the proofs are different in several of the cases. 
New insights into the discrete time model are presented via simulations. 
Finally, the data results from \cite{usda_acc,dtjournal} are improved upon when using a two-competing virus model.

The paper is organized as follows: in Section \ref{sec:mod}, the competing virus spread model is introduced with accompanying assumptions that ensure the model is well posed, which is proven. 
In Section \ref{sec:analysis}, we analyze the model from Section \ref{sec:mod}. 
In Section \ref{sec:id}, we present necessary and sufficient conditions for learning, or identifying, the spread process parameters of the same model, from data produced by the models. In so doing, we establish several assumptions that need to be met by the USDA data. 
In Section \ref{sec:sim}, we validate the results from Sections \ref{sec:analysis} and \ref{sec:id} via simulation and present some exploratory simulations to support the work in Section \ref{sec:usda} and 
propose an online spread parameter estimation algorithm.
In Section \ref{sec:usda}, we learn the spread parameters of the USDA subsidy programs using data from different subsets of the country 
and verify the learned parameters by simulating the spread model over the complete 
United States and comparing the simulated data with the actual data. 
\subsection{Notation}

Given a vector function of continuous time $x$, $\dot{x}$ indicates the time-derivative. Given a vector function of discrete time $x[t]$, $t$ is the time index.
Given a vector $x \in \mathbb{R}^{n}$, the 2-norm is denoted by $\|x\|$ and the transpose by $x^{\top}$. 
The vector of all equal zeros is denoted by $\0$. 
Given two vectors $x^1,x^2\in \mathbb{R}^{n}$, $x^1 > x^2$ indicates each element of $x^1$ is greater than or equal to the corresponding element of $x^2$ and $x^1 \neq x^2$, and $x^1 \gg x^2$ indicates each element of $x^1$ is strictly greater than the corresponding element of $x^2$. 
Given a matrix $A \in \mathbb{R}^{n \times n}$, the spectral radius is $\rho(A)$.
Also, $a_{ij}$ indicates the $i, j^{th}$ entry of the matrix $A$, and $\| A \|_{F}$ indicates the  Frobenius norm of $A$. 
The notation $diag(\cdot)$ refers to a diagonal matrix with the argument(s) on the diagonal; the argument can be a vector $x$ or its elements $x_i$. 
For $n\in \Z^+$, $[n] := \{ 1, ... , n\}$.

\section{Competing Virus Model}\label{sec:mod}

We introduce a discrete-time multi-virus competing model. 
The model can be derived from the continuous-time model, where, for each virus $k\in [m]$, $x^k_i$ is the infection level of the $i$th agent (which can be interpreted as the probability of agent $i$ being infected or the proportion of subpopulation $i$ that is infected) and evolves as
\begin{equation}\label{eq:cont}
    \dot{x}_{i}^{k} =  (1-x^{1}_i - \dots - x^{m}_i) \sum^{n}_{j=1} \beta^k_{ij}x^k_{j}- \delta_i x^k_{i},
\end{equation}
where 
$\beta^k_{ij}>0$ are the infection rates and non-negative, edge weights between the agents/groups and $\delta^k_i>0$ is the healing rate, both associated with virus $k\in [m]$ and for agent $i$. 
Applying Euler's method \cite{atkinson2008introduction} to \eqref{eq:cont} gives 
\begin{align}
    x_{i}^{k}[t+1] = x_{i}^{k}[t] + h\left((1-x^{1}_i[t] - \dots - x^{m}_i[t]) \sum^{n}_{j=1} \beta^k_{ij}x^k_{j}[t]  - \delta^k_i x^k_{i}[t] \right),\label{eq:dis}
\end{align}
where $t$ is the time index and $h>0$ is the sampling parameter. 
We can write \eqref{eq:dis} in matrix form 
\begin{equation}\label{eq:disG}
     x^{k}[t+1] = x^k[t] + h((I-X^1-\cdots - X^m)B^k-D^k)x^k[t],
\end{equation}
where $X^k = diag(x^k[t])$, $B^k$ is the matrix of $\beta^k_{ij}$, and $D^k = diag(\delta^k_i)$. 
Note that $B^k$ is not symmetric in general. 
For the model to be well defined we introduce several assumptions.
\begin{assumption}
For all $i\in[n]$ and $k\in [m]$, we have $x^{k}_i[0],(1-x^{1}_i[0] - \dots - x^{m}_i[0])\in[0,1]$.
\label{x0}
\end{assumption}

\begin{assumption}
For all $i\in[n]$ and $k\in [m]$, we have $\delta^k_i\geq0$ and, for all $j\in[n]$, $\beta^k_{ij} \geq 0$.
\label{pos}
\end{assumption}

\begin{assumption}
For all $i\in[n]$ and $k\in [m]$, we have $h\delta^k_i\leq 1$ and $h\sum_{k = 1} ^ {m} \sum_{j = 1} ^ { n}\beta^k_{ij}\leq 1$.
\label{01}
\end{assumption}


\begin{lemma}
For the system in \eqref{eq:disG}, under the conditions of Assumptions  \ref{x0}, \ref{pos}, and \ref{01}, $x^{k}_i[t],(1-x^{1}_i[t] - \dots - x^{m}_i[t])\in[0,1]$ for all $i\in[n]$, $k\in[m]$, and $t\ge 0$.
\label{lem:box}
\end{lemma}
\begin{proof}
Suppose that at some time $t \geq 0$, $x^{k}_i[t],(1-x^{1}_i[t] - \dots - x^{m}_i[t])\in[0,1]$ for all $i\in[n]$ and $k\in[m]$. 
Consider an arbitrary node $i\in[n]$.  
Summing \eqref{eq:dis} over $k$ and rearranging terms gives
\begin{align}\nonumber
\sum_{k=1}^m x_i^k[t+1] 
&= 
(1-\sum_{k=1}^m x_i^k[t])h \sum_{k=1}^m\sum^{n}_{j=1} \beta_{ij}^kx_{j}^{k}[t] 
+ \sum_{k=1}^m x_i^k[t] (1-h\delta_i^k) \nonumber\\
&\leq  (1-\sum_{k=1}^m x_i^k[t])h \sum_{k=1}^m \sum^{n}_{j=1} \beta_{ij}^k + \sum_{k=1}^m x_i^k[t] (1-h\delta_i^k) \label{eq:xleq1}\\
&\leq 1, \label{eq:leq1}
\end{align}
where \eqref{eq:xleq1} holds since $x_{j}^{k}[t]\leq 1$ for all $j\in[n]$ and $k\in[m]$ and \eqref{eq:leq1} holds since \eqref{eq:xleq1} is a convex combination of $h\sum_{k=1}^m \sum^{n}_{j=1} \beta_{ij}^k$ and $(1-h\delta_i^k)$, which are less than or equal to one by Assumption \ref{01}. Therefore $(1-x^{1}_i[t+1] - \dots - x^{m}_i[t+1])\geq 0$. 

Consider 
an arbitrary virus $k\in[m]$.
Since $(1-x^{1}_i[t] - \dots - x^{m}_i[t]) \sum^{n}_{j=1} \beta^k_{ij}x^k_{j}[t]\geq 0$,
 we have, from~\eqref{eq:dis}, 
\begin{align}\nonumber
    x_{i}^{k}[t+1] \geq ( 1 -h \delta^k_i ) x^k_{i}[t]\geq 0,\nonumber
\end{align}
by Assumption \ref{01}. 
Therefore $\sum_{k=1}^m x_i^k[t+1] \geq 0$ and thus $(1-x^{1}_i[t+1] - \dots - x^{m}_i[t+1])\leq 1$. Consequently we have shown $(1-x^{1}_i[t+1] - \dots - x^{m}_i[t+1])\in [0,1]$.

By rearranging \eqref{eq:dis}, we have
\begin{align}
x_{i}^{k}[t+1] &= x_{i}^{k}[t](1- h\delta_i^k) + (1-x_{i}^{k}[t])\left(h \sum^{n}_{j=1} \beta_{ij}^kx_{j}^{k}[t]\right) -\sum_{l\neq k}x_{i}^{l}[t]\left(h\sum^{n}_{j=1} \beta_{ij}^kx_{j}^{k}[t]\right) \label{eq:neg}
\\
&\leq \underbrace{x_{i}^{k}[t](1- h\delta_i^k) + (1-x_{i}^{k}[t])\left(h \sum^{n}_{j=1} \beta_{ij}^k 
\right)}_{z_{i}^{k}[t]},
\nonumber
\end{align}
since the term on line \eqref{eq:neg} is non-positive. 
Since $x^{k}_i[t]\in [0,1]$, $z_{i}^{k}[t]$ is a convex combination of $(1- h\delta_i^k)$ and $h \sum^{n}_{j=1} \beta_{ij}^k$, which are less than or equal to one by Assumption \ref{01}, $z_{i}^{k}[t]\leq 1$, which implies $x_{i}^{k}[t+1]\leq 1$. Consequently we have shown $x^{k}_i[t+1] \in [0,1]$.

Further, by Assumption \ref{x0}, $x^{k}_i[0],(1-x^{1}_i[0] - \dots - x^{m}_i[0])\in[0,1]$ for all $i\in[n]$ and $k\in [m]$, thus
it follows that $x^{k}_i[t],(1-x^{1}_i[t] - \dots - x^{m}_i[t])\in[0,1]$ for all $i\in[n]$, $k\in [m]$, and $t\ge 0$.
\end{proof}
Lemma \ref{lem:box} implies that the set 
\begin{equation}\label{D}
    \D     =\left\{(x^1,\dots , x^m) \; | \; x^k \geq \0, \ k\in [m], \; \sum_{k=1}^m x^k\leq \1 \right\}
\end{equation}
is positively invariant with respect to the system defined by \eqref{eq:disG}. Since $x^k_i$ 
denotes the probability of infection of individual $i$ by virus $k$, or the fraction  of group $i$ infected by virus $k$, 
and $1-x^1_i-\cdots - x^m_i$ denotes the probability of individual $i$ being healthy, or the fraction of group $i$ that is healthy, it is natural to assume that their initial values are in the interval $[0,1]$, 
since otherwise the values will lack any physical meaning for the epidemic model considered here. 
Therefore, 
we focus on the analysis of \eqref{eq:disG} only on the domain~$\D$.

We need an assumption to ensure \textit{non-trivial} 
virus spread.
\begin{assumption}\label{not0}
 We have $B\neq 0$, $h\neq 0$, and $n>1$.
\end{assumption}

\section{Analysis}
\label{sec:analysis}

\begin{definition}
Consider an autonomous system
\begin{equation}
     x[t+1]  = f(x[t]), \label{def}
\end{equation}
where $f: \scr{X}\rightarrow\R^
n$ is a locally Lipschitz map from a domain $\scr{X}\subset\R^
n$ into $\R^n$. 
Let $\tilde{x}$ be an equilibrium of \eqref{def} and $\scr{E}\subset\scr{X}$ be a domain containing $\tilde{x}$.
If the equilibrium $\tilde{x}$ is asymptotically stable such that for any $x[0]\in\scr{E}$ we have
$\displaystyle\lim_{t\rightarrow\infty}x[t] = \tilde{x}$, then $\scr{E}$ is said to be a domain of attraction for~$\tilde{x}$.
\end{definition}

\begin{prop}
Let $\tilde{x}$ be an equilibrium of \eqref{def} and $\scr{E}\subset \scr{X}$ be a domain
containing $\tilde{x}$. Let $V:\scr{E}\rightarrow\R$ be a continuously differentiable function
such that $V(\tilde{x})=\tilde{x}$, $V(x)>0$ for all $x$ in $\scr{E}\setminus \{\tilde{x}\}$, 
and $\Delta V[t] := V(x[t+1]) - V(x[t])<0$ for all $x[t]$ in $\scr{E}\setminus \{\tilde{x}\}$. If $\scr{E}$ is a positively invariant set,
then the equilibrium $\tilde{x}$ is asymptotically stable with a domain of attraction $\scr{E}$.
\label{prop:lya}
\end{prop}
\noindent This proposition is a direct consequence of Lyapunov's stability theorem for discrete-time systems
and the definition of domain of attraction.

Finally, we need an assumption on the structure of the $B^k$ matrices. A square matrix is called {\em irreducible}  if it cannot be permuted to a block upper triangular matrix.
\begin{assumption}
For all $k\in [m]$, $B^k$ is irreducible.
\label{connect}
\end{assumption}
\noindent Note that this assumption is equivalent to the underlying graph being strongly connected. We have the following result about the healthy state, where $x^k_i = 0$ for all $i\in[n]$, $k\in [m]$.

\begin{theorem} \label{thm:0global}
Suppose that Assumptions \ref{x0}-\ref{connect} hold for \eqref{eq:disG}. If $\rho(I-hD^k+hB^k)\leq 1$ for all $k\in [m]$, then the healthy state is asymptotically stable with domain of attraction $\D$, as defined in \eqref{D}.
\end{theorem}


\begin{proof}
We employ a LaSalle's invariance principle argument.
To simplify notation, let $M^k = I+hB^k-hD^k$, $\hat{X} = (X^1[t]+\dots + X^m[t])$, and $\hat{M}^k = I + h((I-\hat{X})B^k-D^k)$. By Assumptions \ref{pos}-\ref{connect}, $M^k$ is an irreducible nonnegative matrix. First we evaluate the case where $\rho(I-hD^k+hB^k)< 1$ for all $k\in [m]$.  Therefore, by Proposition 1 in \cite{rantzer2011positive},  for all $k\in [m]$, there exists a positive diagonal matrix $P^k_1$ such that $(M^k)^{\top} P^k_1 M^k - P^k_1$ is negative definite. Consider the Lyapunov function $V_1^k(x^k[t]) = (x^k[t])^{\top} P^k_1 x^k[t]$. 
For each $k\in [m]$, using \eqref{eq:disG} with $x^k[t]\neq 0$ (dropping the $[t]$ for notation convenience) gives
\begin{align}
  \Delta V_1^k[t] &= (x^{k})^{\top}(\hat{M}^k)^{\top} P^k_1 (x^{k})^{\top}\hat{M}^k x^{k} - (x^k)^{\top}P^k_1 x^k \nonumber \\
  &= (x^{k})^{\top}((M^k)^{\top} P^k_1 M^k - P^k_1) x^{k} 
  - 2h(x^{k})^{\top}(B^k)^{\top}\hat{X}P^k_1 M^k x^k 
  + h^2(x^{k})^{\top}(B^k)^{\top}\hat{X}P^k_1 \hat{X} B^k x^k \nonumber \\ 
  &< h^2(x^{k})^{\top} (B^k)^{\top} \hat{X}P^k_1 \hat{X} B^k x^k 
  - 2h(x^{k})^{\top}(B^k)^{\top}\hat{X}P^k_1 M^k x^k \label{eq:pos} \\ 
  &= h^2(x^{k})^{\top}(B^k)^{\top}\hat{X}P^k_1 \hat{X} B^k x^k  
  - 2h^2(x^{k})^{\top}(B^k)^{\top}\hat{X}P^k_1 (B^k)^{\top} x^k 
  - 2h(x^{k})^{\top}(B^k)^{\top}\hat{X}P^k_1 (I-hD) x^k  \nonumber \\ 
  &\leq h^2((x^{k})^{\top}(B^k)^{\top}\hat{X}P^k_1 \hat{X} B^k x^k 
  - 2(x^{k})^{\top}(B^k)^{\top}\hat{X}P^k_1 (B^k)^{\top} x^k) \label{eq:hD} \\
  &\leq -h^2(x^{k})^{\top}(B^k)^{\top}\hat{X}P^k_1(I - \hat{X}) B^k x^k \nonumber \\
  &\leq 0, \label{eq:0}
\end{align}
where \eqref{eq:pos} holds by Proposition 1 in \cite{rantzer2011positive}, \eqref{eq:hD} holds by Assumptions \ref{pos} and \ref{01}, and \eqref{eq:0} holds by Lemma \ref{lem:box}. 
Therefore, by Proposition~\ref{prop:lya}, $x^k$ converges asymptotically to the origin. Since $k \in [m]$ was chosen arbitrarily, the whole system, that is, every virus $k \in [m]$, converges to the healthy state for this case.

For the case where $\rho(I-hD^k+hB^k)= 1$, we have, by Lemma 3 in \cite{dtjournal}, that,  for all $k\in [m]$, there exists a positive diagonal matrix $P^k_2$ such that $(M^k)^{\top} P^k_2 M^k - P^k_2$ is negative semi-definite. Consider the Lyapunov function $V_2(x^k) = (x^k)^{\top} P^k_2 x^k$. Using \eqref{eq:disG} with $x^k\neq 0$, gives
\begin{align}
  \Delta V_2[k] &= (x^{k})^{\top}(\hat{M}^k)^{\top} P^k_2 (x^{k})^{\top}\hat{M}^k x^{k} - (x^k)^{\top} P^k_2 x^k \nonumber \\
  &= (x^{k})^{\top}((M^k)^{\top} P^k_2 M^k - P^k_2) x^{k} 
  - 2h(x^{k})^{\top}(B^k)^{\top}\hat{X}P^k_2 M^k x^k 
  + h^2(x^{k})^{\top} (B^k)^{\top} \hat{X}P^k_2 \hat{X} B^k x^k \nonumber \\ 
  &< h^2(x^{k})^{\top}(B^k)^{\top}\hat{X}P^k_2 \hat{X} B^k x^k
  - 2h(x^{k})^{\top}(B^k)^{\top}\hat{X}P^k_2 M^k x^k \nonumber \\ 
  &= h^2(x^{k})^{\top}(B^k)^{\top}\hat{X}P^k_2 \hat{X} B^k x^k 
  - h(x^{k})^{\top}(B^k)^{\top}\hat{X}P^k_2 M^k x^k 
  -h^2(x^{k})^{\top}(B^k)^{\top}\hat{X}P^k_2 B^k x^k 
  \nonumber \\
  & \ \ \ \ \ \ 
  - h(x^{k})^{\top}(B^k)^{\top}\hat{X}P^k_1 (I-hD) x^k  \nonumber \\ 
  &\leq h^2(x^{k})^{\top}(B^k)^{\top}\hat{X}P^k_2 \hat{X} B^k x^k  
  - h(x^{k})^{\top}(B^k)^{\top}\hat{X}P^k_2 M^k x^k 
  -h^2(x^{k})^{\top}(B^k)^{\top}\hat{X}P^k_2 B^k x^k  \nonumber \\ 
  &\leq h^2(x^{k})^{\top}(B^k)^{\top}\hat{X}P^k_2 (I - \hat{X}) B^k x^k  
  - h(x^{k})^{\top}(B^k)^{\top}\hat{X}P^k_2 M^k x^k  \nonumber  \\
  &\leq - h(x^{k})^{\top}(B^k)^{\top}\hat{X}P^k_2 M^k x^k \nonumber \\
  &\leq 0. \nonumber
\end{align}
Clearly if $x^k = \0$, then $- h(x^{k})^{\top}(B^k)^{\top}\hat{X}P^k_2 M^k x^k =~0$. Since, by Assumptions \ref{pos} and \ref{not0} and by Lemma 3 in~\cite{dtjournal}, $B^k,M^k,P^k_2$ are nonzero, nonnegative matrices, if $- h(x^{k})^{\top}(B^k)^{\top}\hat{X}P^k_2 M^k x^k = 0$, then $x^k = \0$. Therefore, by Proposition \ref{prop:lya}, $x^k$ converges asymptotically to the origin. Since $k \in [m]$ was chosen arbitrarily, the whole system, that is, every virus $k \in [m]$, converges to the healthy state for this case. Therefore the healthy state is asymptotically stable with  domain of attraction $\D$.
\end{proof}

\begin{prop}\label{prop:endemic}
   Suppose that Assumptions \ref{x0}-\ref{connect} hold. If $\rho(I-hD^k+hB^k)> 1$ for all $k\in [m]$, then \eqref{eq:disG} has at least $k+1$ equilibria, $\0$, $(\tilde{x}^1, \0, \dots, \0)$, $\dots$, $(\0, \dots, \0, \tilde{x}^m)$, where, for each $k\in[m]$, $\tilde{x}^k\gg \0$.
\end{prop}
\begin{proof}
  Clearly $\0$ is always an equilibrium of \eqref{eq:disG}.
  
  By the Perron Frobenius Theorem for irreducible nonnegative matrices (Theorem 8.4.4 in \cite{horn2012matrix}), for all $k\in [m]$, $\rho(I-hD^k+hB^k) = s_1(I-hD^k+hB^k)$ and there exists $v^k \gg 0$ such that $$(I-hD^k+hB^k)v^k = \rho(I-hD^k+hB^k)v^k > v^k,$$
since $\rho(I-hD^k+hB^k)>1$. 
Therefore 
  \begin{align*}
      (-hD^k+hB^k)v^k &= \rho(-hD^k+hB^k)v^k 
      = s_1(-hD^k+hB^k)v^k > 0v^k,
  \end{align*}
which implies 
  $$ \rho(I-hD^k+hB^k)> 1 \Longleftrightarrow h(s_1(-D^k+B^k))> 0.$$
  This condition is the same as the condition of Proposition 3 in \cite{liu2016onthe}, which shows the existence and stability of the endemic state in the single virus case. The proof follows similarly, when assuming that $x^l = \0$ for all $l\neq k$, that there exists $\tilde{x}^k\gg \0$ such that 
  $$h((-D^k+B^k)-\tilde{X}^kB)\tilde{x}^k=\0.$$
  Therefore, if $x^l = \0$ for all $l\neq k$, $\tilde{x}^k$ is an equilibrium of \eqref{eq:disG}.  
  Consequently, $\0$, $(\tilde{x}^1, \0, \dots, \0)$, $\dots$, $(\0, \dots, \0, \tilde{x}^m)$ are equilibria of \eqref{eq:disG}.
\end{proof}

We have the following corollary. 
\begin{corollary}\label{cor:1survive}
 Suppose that Assumptions \ref{x0}-\ref{connect} hold. If $\rho(I-hD^k+hB^k)\leq 1$ for all $k\in [m]\setminus \{l\}$ and $\rho(I-hD^l+hB^l)> 1$, then \eqref{eq:disG} has two equilibria $\0$ and $(\0, \dots, \0,  \tilde{x}^l, \0, \dots, \0)$ with $\tilde{x}^l\gg \0$. Furthermore, $\0$ is asymptotically stable with domain of attraction equal to $\{(x^1,\dots,x^m )| x^l = \0 \text{ and } x^k \in [0,1]^n \ \forall k \neq l\}$ and $(\0, \dots, \0,  \tilde{x}^l, \0, \dots, \0)$ is  locally asymptotically stable. 
\end{corollary}
\begin{proof}
  The existence of the equilibria follows from Theorem \ref{thm:0global} and Proposition \ref{prop:endemic}. The asymptotically stability of $\0$ with domain of attraction equal to $\{(x^1,\dots,x^m )| x^l = \0 \text{ and } x^k \in [0,1]^n \ \forall k \neq l\}$ follows directly from Theorem \ref{thm:0global} since virus $l$, which is the only virus with $\rho(I-hD^l+hB^l)> 1$, is always equal to zero.
  
  From the proof of Theorem~\ref{thm:0global}, $x^{k}[t]$ will asymptotically converge to $\0$
as $t\rightarrow\infty$ for all initial values $(x^{1}[0],\dots x^m[0]) \in \{(x^1,\dots,x^m )| x^l = \0 \text{ and } x^k \in [0,1]^n \ \forall k \neq l\}$, for $k\neq l$.
From \eqref{eq:disG},
$$x^{l}[t+1] = x^l[t] + h( B^{l}- D^{l} - X^{l}B^{l}) x^{l}[t] - h\sum_{k\neq l} X^{k}B^{k}x^{k}[t].$$
Thus, we  regard the dynamics of $x^{l}[t]$ as an autonomous system
\begin{equation}
    x^{l}[t+1] = x^l[t] + h(B^{l}- D^{l} - X^{l}(t)B^{l}) x^{l}(t),\label{temp}
\end{equation}
with a vanishing perturbation $- h\sum_{k\neq l} X^{k}(t)B^{l}x^{l}(t)$, which converges to $\0$  as $t\rightarrow\infty$.
Therefore, from Theorem 2 in \cite{prasse2019viral}, 
the autonomous system \eqref{temp} will locally asymptotically converge to
the unique epidemic state $\tilde x^{l}$. Thus \eqref{eq:disG} will converge to $(\0, \dots, \0, \tilde x^{l}, \0, \dots, \0)$. 
\end{proof}

From Theorem \ref{thm:0global} and Proposition \ref{prop:endemic}, we have the following result.

\begin{theorem}
Under Assumptions \ref{x0}-\ref{connect}, the healthy state is the unique equilibrium of \eqref{eq:disG} if and only if
 $\rho(I-hD^k+hB^k)\leq~1$ for all $k\in [m]$.
\label{thm:eq0}
\end{theorem}



\section{Learning Spread Parameters}\label{sec:id}

In this section, we clearly lay out the assumptions and the identification techniques for the multi-virus model. For this section we use a slightly different version of \eqref{eq:disG}, where we factor $\beta^k_{ij}$ into $\beta^k_i a^k_{ij}$ as
\begin{equation}\label{eq:disM}
     x^{k}[t+1] = x^k[t] + h((I-X^1-\cdots - X^m)B^kA^k-D^k)x^k[t],
\end{equation}
\vspace{-4ex}

\noindent
where $B^k = diag(\beta^k_i)$ and $A^k$ is the matrix of $a^k_{ij}$. 
\begin{remark}
If the system has homogeneous spread parameters, that is, $\beta^k_i = \beta^k_j$ and $\delta^k_i = \delta^k_j$ for all $i,j\in [n]$, the condition in Theorems \ref{thm:0global}-\ref{thm:eq0} reduces to $\rho(A) \leq \frac{\delta^k}{\beta^k}$.
\label{rem:thres}
\end{remark}
We start by assuming that the underlying graph structures $A^k$ are known and that we have full-state measurement with no noise on the measurements, which we admit are strong assumptions. 
However, for the dataset used in Section \ref{sec:usda} 
these assumptions are well-founded  because we aggregate the data by county and the adjacency of counties is known, i.e., the graph structure is known, and  any farmer that received a subsidy payout is in the dataset, i.e., there are no hidden, unmeasured states. 
We will relax the no-noise assumption in the Simulations Section (see Section~\ref{sec:sim}).


We present several results on learning the spread parameters of the model in \eqref{eq:dis} from data. The following result is an improvement of Theorem 3 in \cite{dtjournal}.
\begin{theorem}\label{thm:idhomo}
Consider the model in \eqref{eq:disG} under Assumptions \ref{x0}-\ref{connect} with 
 virus $k$ having homogeneous spread, that is, 
 $\beta^k$ and $\delta^k$ are the same for all agents. 
Assume that  $A^k$, $x^k[t]$, for all $t \in [T] \cup \{0\}, k\in [m]$,  and $h$ are known. Then, $\beta^k$ and $\delta^k$ can be identified uniquely if and only if $T>0$, and 
there exist $i,j \in [n]$ and $ t_1, t_2\in[T-1] \cup \{0\}$ such that
\begin{align}
x_i^{k}[t_1] g_j(x^{k}[t_2])
\label{eq:inv1}
\neq x_j^{k}[t_2] g_i(x^{k}[t_1]), 
\end{align}
where $g(x^{k}[t]):=(I-X^1[t]-\cdots - X^m[t])A^k x^k[t]$.
\end{theorem}
{\em Proof:}
Since $x^k[0], \dots, x^k[T-1]$,  and $A^k$ are known, using 
\eqref{eq:disM} we can construct the matrix $\Phi^k$, defined as,
\begin{equation}\label{eq:phi}\footnotesize
     \begin{bmatrix} (I-X^1[0]-\cdots - X^m[0])A^k x^k[0] & -x^k[0]\\ \vdots & \vdots \\ (I-X^1[T-1]-\cdots - X^m[T-1])A^k x^k[T-1] & -x^k[T-1] \end{bmatrix}.
\end{equation}\normalsize
Therefore, since we also know $x^k[T]$ and $h$, we can rewrite \eqref{eq:disG} as 
\begin{equation}\label{eq:id1}
    \begin{bmatrix} x^k[1]-x^k[0]\\ \vdots \\ x^k[T]-x^k[T-1] \end{bmatrix} =  h \Phi^k \begin{bmatrix} \beta^k \\ \delta^k \end{bmatrix}.
\end{equation}
Since $n>1$, $\Phi^k$ has at least two rows. 
By the assumption that there exist $i,j \in [n]$ and  $ t_1, t_2\in[T-1] \cup \{0\}$ such that 
\eqref{eq:inv1} holds, 
$\Phi^k$ has 
column rank equal to two, with two unknowns. 
Therefore there exists a unique solution to \eqref{eq:id1} using the inverse or pseudoinverse.

If there do not exist $i,j \in [n]$ and  $ t_1, t_2\in[T-1] \cup \{0\}$ such that \eqref{eq:inv1} holds, then $\Phi^k$ has a nontrivial nullspace. Therefore \eqref{eq:id1} does not have a unique solution. 
\hfill
$\qed$

Note that $t_1$ and $t_2$ from Theorem \ref{thm:idhomo} could both equal zero and the condition in \eqref{eq:inv1} could still hold, that is, recovery of the spread parameters may be possible with only two time series points. 
Now we present two corollaries where $h\beta^k$ and $h\delta^k$, denoted by $\beta^k_h$ and $\delta^k_h$, respectively, can be recovered. 
\begin{corollary}\label{cor:h}
 Consider the model in \eqref{eq:disG} under Assumptions \ref{x0}-\ref{connect} 
 with homogeneous virus spread. 
   Assume that $A^k$ and $x^k[0], \dots, x^k[T]$ are known. Then, $\beta^k_h$ and $\delta^k_h$ can be identified uniquely for every $k \in [m]$ if and only if $T>0$ and 
there exist $i,j \in [n]$ and $ t_1, t_2\in[T-1] \cup \{0\}$ such that
$
x_i^{k}[t_1] g_j(x^{k}[t_2])
\neq x_j^{k}[t_2] g_i(x^{k}[t_1]). $ 
\end{corollary}


This corollary illustrates that under certain conditions, while the exact behavior of the system may not be recoverable 
the limiting behavior of the system may be determined, by employing Theorems \ref{thm:0global}-\ref{thm:eq0} with Remark \ref{rem:thres}.

If the assumption is made that the underlying spread process is heterogeneous, 
we have a similar condition, an improvement of Theorem 4 in \cite{dtjournal}. 

\vspace{-1ex}

\begin{theorem} \label{thm:idhetero}
Consider the model in \eqref{eq:dis} under Assumptions \ref{x0}-\ref{connect}. 
Assume that $x^k[t]$, $x_i^l[t]$ for all $ t\in[T-1] \cup \{0\}, l\in[m]$, $A^k$, $x^k_i[T]$, and $h$ are known. Then, the spread parameters of virus $k$ for node $i$ can be identified uniquely if and only if $T>1$, and 
there exist $ t_1, t_2\in[T-1] \cup \{0\}$ such that
\vspace{-2ex}
\begin{align}
x_i^{k}[t_1] &(1-x_{i}^{1}[t_2]-\cdots - x_i^m[t_2])\sum^{n}_{j=1} a^k_{ij}x_{j}^{k}[t_2] \label{eq:inv}
\neq x_i^{k}[t_2] (1-x_{i}^{1}[t_1]-\cdots - x_i^m[t_1])\sum^{n}_{j=1} a^k_{ij}x_{j}^{k}[t_1] . 
\end{align}
\end{theorem}
\begin{proof}
Since $x^k[t]$, $x_i^l[t]$ for all $ t\in[T-1] \cup \{0\}, l\in[m]$,  and $A^k$ are known, we can construct the matrix $\Phi^k_i$, defined as,
\vspace{-1ex}
\begin{equation*}\footnotesize
     \begin{bmatrix} \displaystyle (1-x_{i}^{1}[0]-\cdots - x_i^m[0]) \sum^{n}_{j=1} a^k_{ij}x_{j}^{k}[0]  & \displaystyle -x_{i}^{k}[0]\\ \vdots & \vdots \\ 
     \displaystyle (1-x_{i}^{1}[T-1]-\cdots - x_i^m[T-1]) \sum^{n}_{j=1} a^k_{ij}x_{j}^{k}[T-1]  & \displaystyle -x_{i}^{k}[T-1] \end{bmatrix}.
\end{equation*}\normalsize

\vspace{-1ex}

\noindent
Then, since we also know $x_i^k[T]$ and $h$, we have 
\vspace{-1ex}
\begin{equation}\label{eq:id2}
    \begin{bmatrix} \displaystyle x_{i}^{k}[1]-x_{i}^{k}[0]\\ \vdots \\ \displaystyle x_{i}^{k}[T]-x_{i}^{k}[T-1] \end{bmatrix} =  h\Phi_i^k \begin{bmatrix} \beta_{i}^k \\ \delta_{i}^k \end{bmatrix}.
\end{equation}
Since $T>1$, $\Phi^k_i$ has at least two rows. 
By the assumption that there exist $ t_1, t_2\in[T-1] \cup \{0\}$ such that 
\eqref{eq:inv} holds, 
$\Phi^k_i$ has 
column rank equal to two, with two unknowns. 
%
Therefore there exists a unique solution to \eqref{eq:id2} using the inverse or pseudoinverse.

If there do not exist $ t_1, t_2\in[T]$ such that \eqref{eq:inv} holds, then $\Phi^k_i$ has a nontrivial nullspace. Therefore \eqref{eq:id2} does not have a unique solution. 
\end{proof}

\section{Simulations}\label{sec:sim}

In this section, we  present first, a set of simulations that illustrate the results from the previous sections and second, some illuminating simulations of the model that support the validation work with real data. 
Since the dataset we consider in Section \ref{sec:usda} only has two competing spread processes we limit ourselves to $m=2$ for this section as well, however, the behavior is similar for $m>2$.
Virus 1 is depicted by the color red ($r$), virus 2 is depicted by the color green ($g$), and susceptible, or healthy, is depicted by the color blue ($b$). For all $i\in[n]$, the color at each time $t$ for node $i$ is given by 
\begin{equation}\label{eq:color}
x^{1}_i[t]r + x^{2}_i[t]g + (1-x^{1}_i[t]-x^{2}_i[t])b.
\end{equation}
For the second set of simulations we, at times, inspect the case of $m=1$. 

\subsection{Examples of Results}

\begin{figure}
    \centering
    \begin{subfigure}[b]{.32\columnwidth}
      \includegraphics[width=\columnwidth]{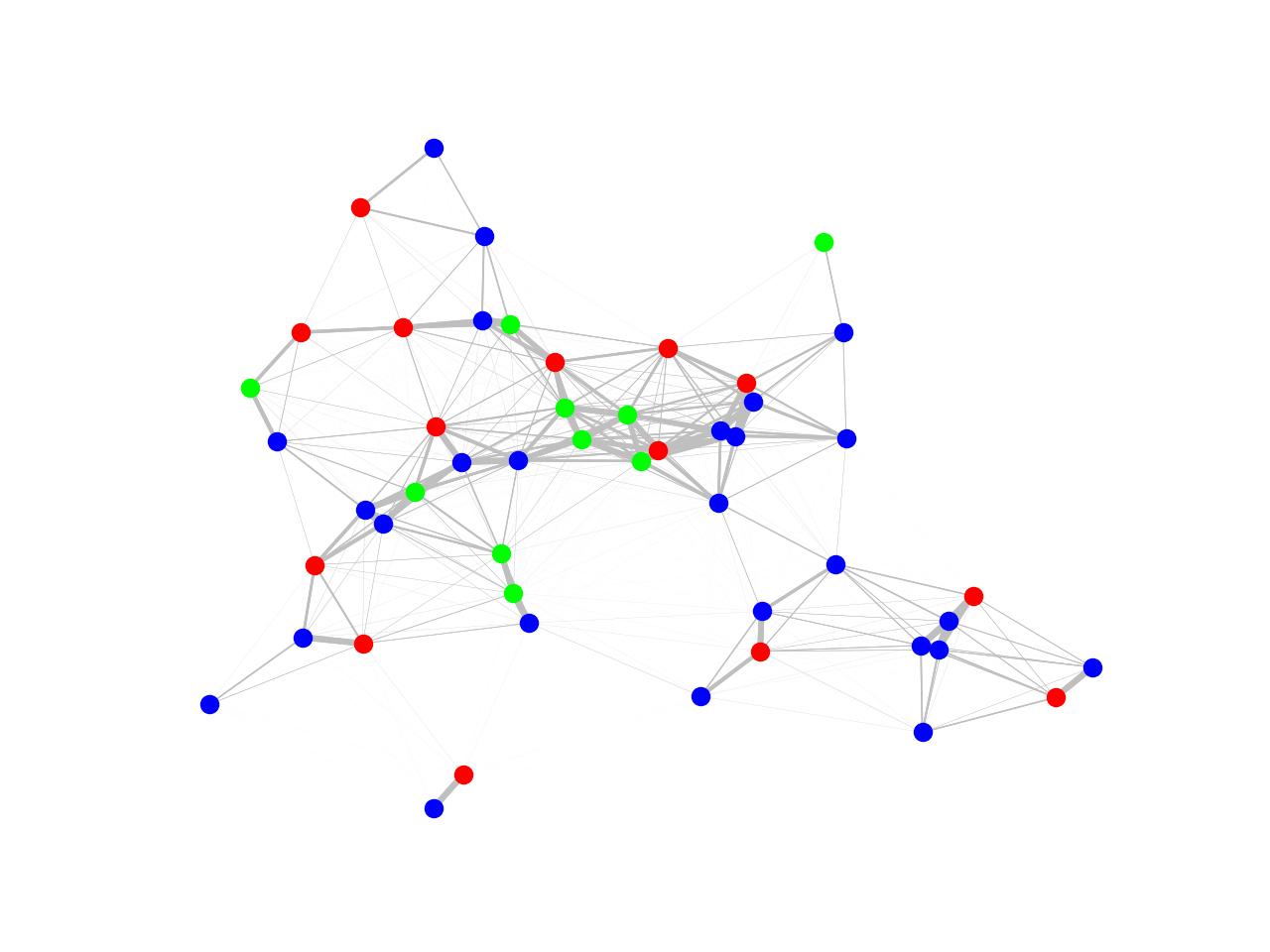}
      \caption{The system at time zero.}
      \label{fig:homo_1}
    \end{subfigure}
    \hfill
    \begin{subfigure}[b]{.32\columnwidth}
      \includegraphics[width=\columnwidth]{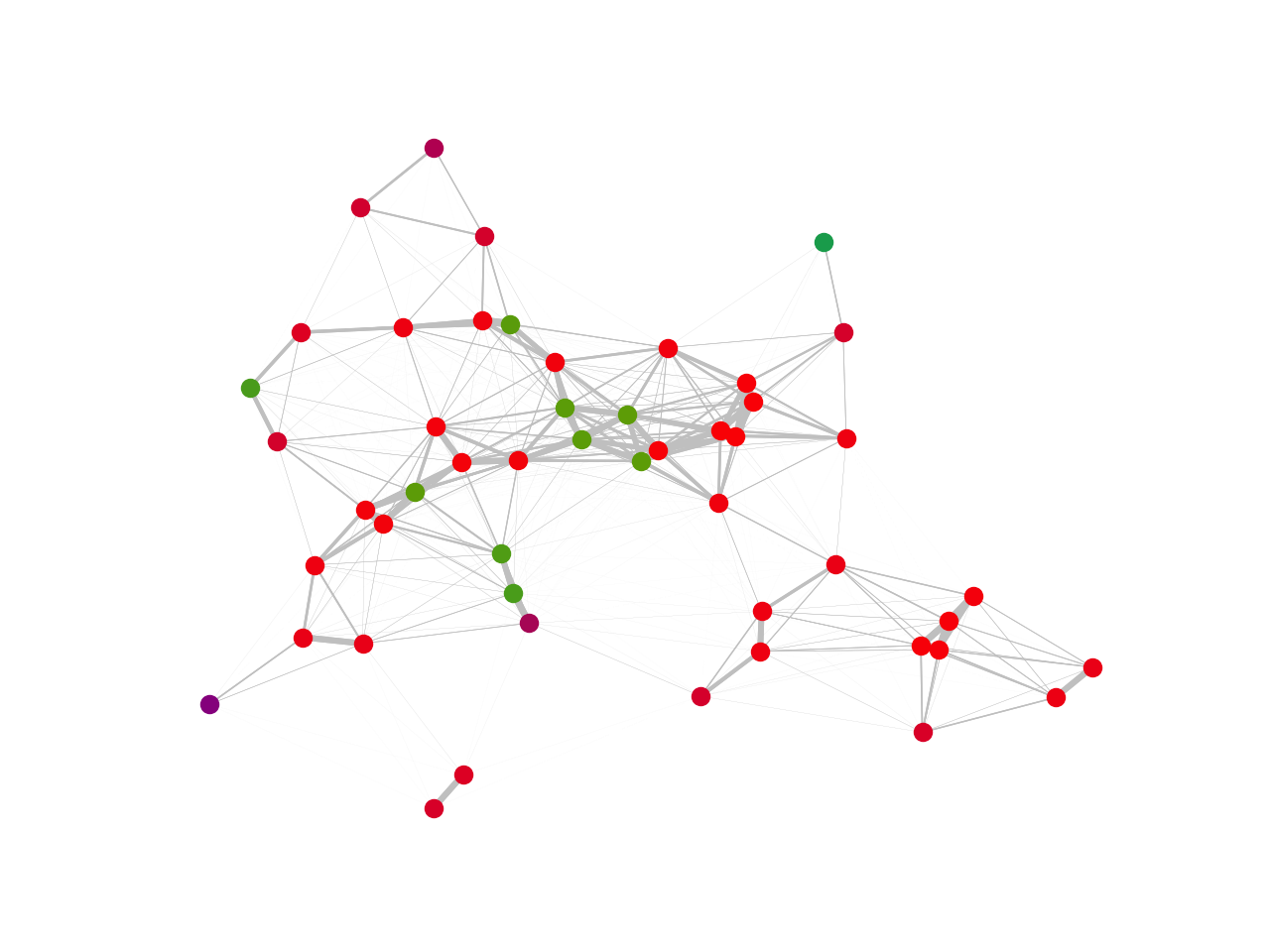}
      \caption{The system at time 100.}
      \label{fig:homo_100}
    \end{subfigure}
    \hfill
    \begin{subfigure}[b]{.32\columnwidth}
      \includegraphics[width=\columnwidth]{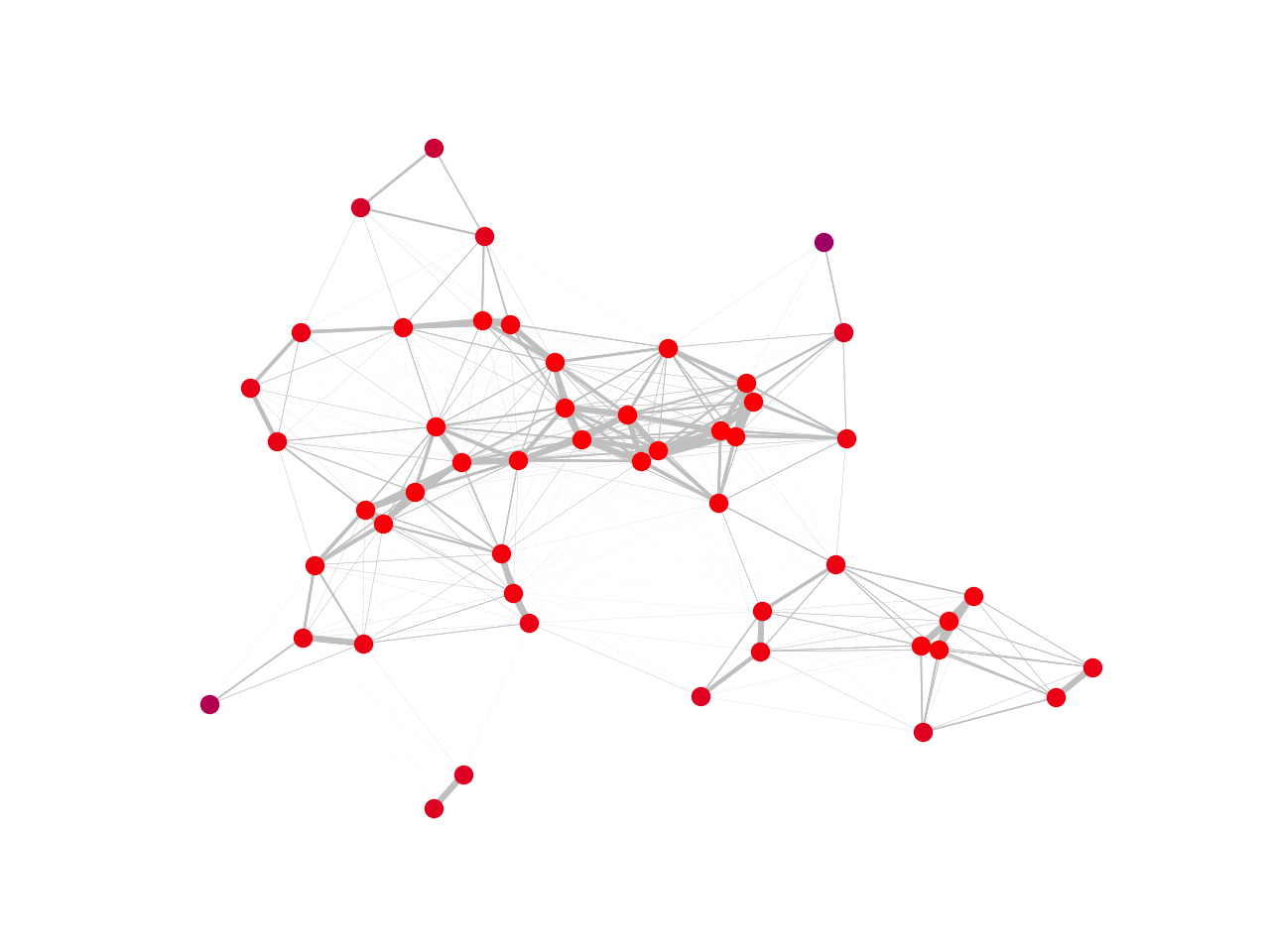}
      \caption{The system at time 10000.}
      \label{fig:homo_10000}
    \end{subfigure}
    \caption{This homogeneous virus system follows \eqref{eq:dis} with $\beta^1 = 1$, $\beta^2 = 0.01$ $\delta^1 = 0.1$, $\delta^2 = 0.1$, $h=0.05$, and $A$ depicted by the edges.
    }
\label{fig:homo}
\end{figure}

We consider a system with $50$ agents $24$ of which are randomly chosen such that they are initially infected by either one of the two competing viruses. For Virus 1, $\beta^1 = 1$ and $\delta^1 = 0.1$ and for the Virus 2, $\beta^2 = 0.01$ and $\delta^2 = 0.1$. Moreover, $h = 0.05$ and the weighted adjacency matrix for both viruses, $A$, is determined by
\begin{equation} \label{eq:A_sim}
    a_{ij} = \begin{cases}
    e^{- \|z_i - z_j\|^2}, & \text{ if } i \neq j, \\ 0, & \text{ otherwise,} \end{cases}
\end{equation}
where $z_i$ is the position of agent $i$ and $A$ is, therefore, fully connected. Since the edges are weighted, the ones between nodes that are far away from each other are difficult to see in the figures. A simulation, based on this system, is shown in Figure~\ref{fig:homo} with plots of the initial condition, the epidemic states at time-step $100$ and the final condition. Assuming $A$ is known we recover $\beta^1_h$, $\delta^1_h$, $\beta^2_h$, and $\delta^2_h$ exactly, using \eqref{eq:id1} with only two time-steps, consistent with Corollary~\ref{cor:h}. Hence, the proportions  $\delta^1/\beta^1$ and  $\delta^2/\beta^2$ are also correctly recovered.
And clearly, if $h$ is known, we recover the parameters exactly, consistent with Theorem \ref{thm:idhomo}. 
Moreover,
\begin{equation*}
    \rho(I - hD^1 + h\beta^1 A) = 1.1976 > 1,  \text{ and } \rho(I - hD^2 + h\beta^2 A) = 0.997 \leq 1,
\end{equation*}
and consistent with Corollary~\ref{cor:1survive}, the endemic state is $(\tilde{x}^1, \0)$, where $\tilde{x} \gg \0$. We also find that this endemic equilibrium is reached for all initial conditions with $x^1[0] > \0$, that is, via simulations it appears to be globally stable.

\begin{figure}
    \centering
    \begin{subfigure}[b]{.493\columnwidth}
      \includegraphics[width=\columnwidth]{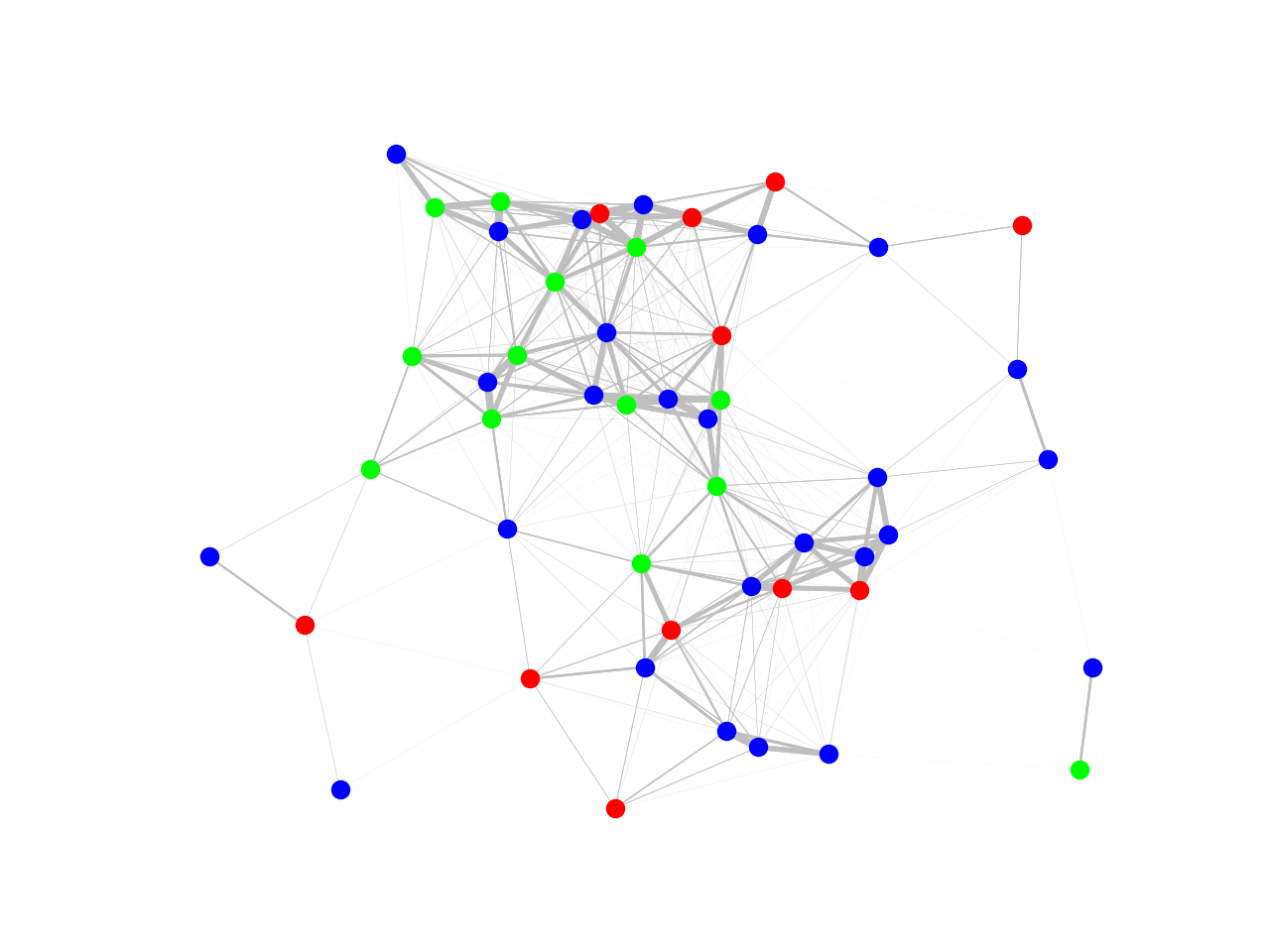}
      \caption{The system at time zero.}
      \label{fig:dis0}
    \end{subfigure}
    \hfill
    \begin{subfigure}[b]{.493\columnwidth}
      \includegraphics[width=\columnwidth]{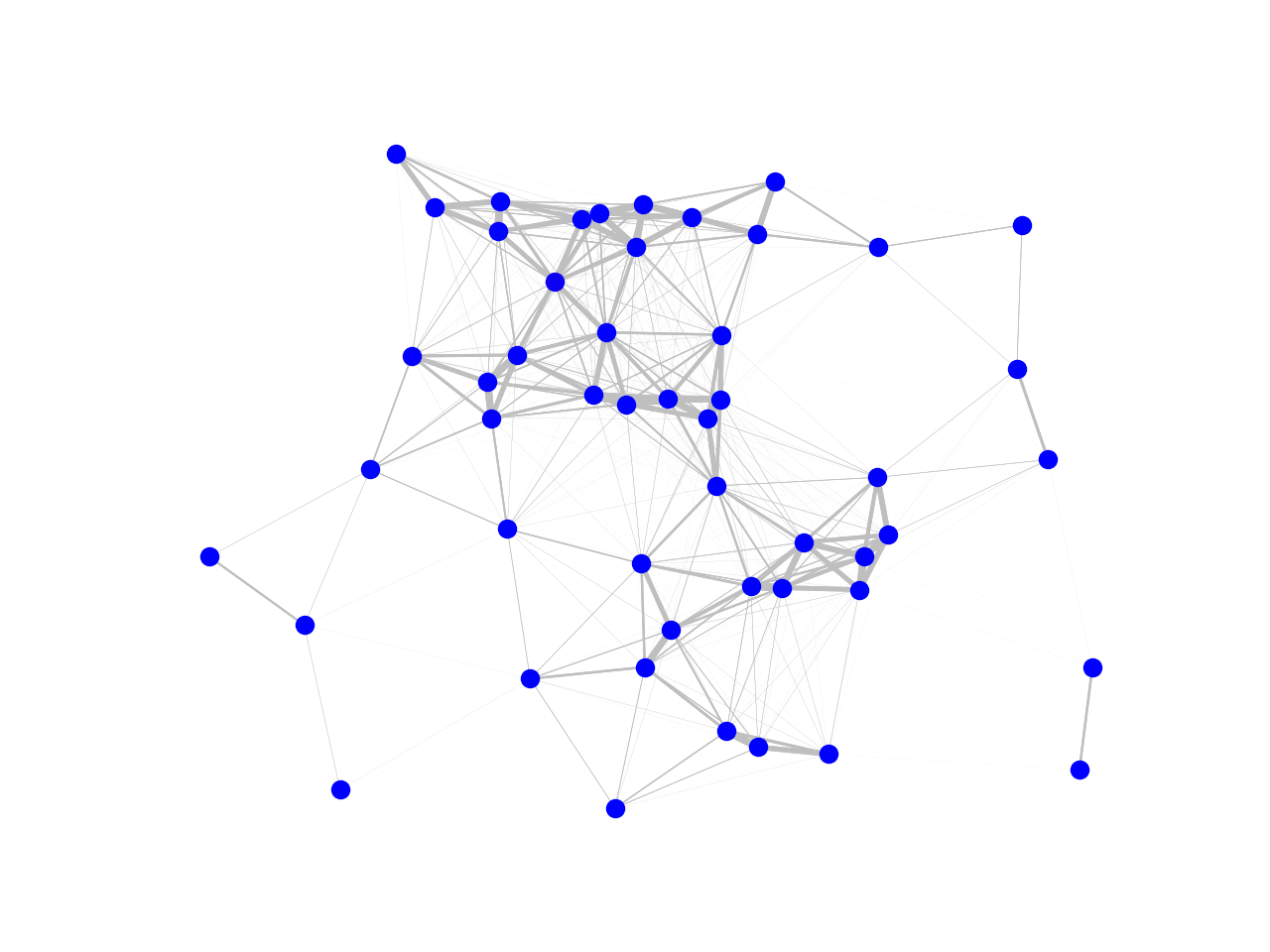}
      \caption{The system at time 100.}
      \label{fig:dis100}
    \end{subfigure}
    \caption{This heterogeneous virus system follows \eqref{eq:dis} with $\beta^k_i \in [0.001, \; 1]$ and $\delta^k_i \in [0.1, \; 10]$ randomly generated $\forall i,k$, $h=0.05$, and $A$ depicted by the edges.
    }
\label{fig:sim_graph}
\end{figure}

We now consider a similar system with $50$ agents $24$ of which are initially infected by either one of the two viruses and $A$ given by \eqref{eq:A_sim}. But the agents have moved and the system is a heterogeneous virus system with $\beta^k_i \in [0.001, \; 1]$ and $\delta^k_i \in [0.1, \; 10]$ randomly generated from uniform distributions for all $i\in[n]$ and $k\in[m]$. For $T = 3$ the assumptions in Theorem~\ref{thm:idhetero} are met and we recover the spread parameters exactly.
Moreover, 
\begin{equation*}
    \rho(I - hD^1 + h\beta^1 A) = 0.9958 \leq 1,  \text{ and }\rho(I - hD^2 + h\beta^2 A) = 0.9851 \leq 1,
\end{equation*}
and we observe that the system converges to the healthy state, $\tilde{x} = \0$, consistent with Theorem~\ref{thm:0global}.

\subsection{Exploratory Simulations} \label{sec:exp}

In this section, we present two set of simulations that give important insight into the model to assist our work on the USDA dataset in the next section. The first simulation explores how accurately we can capture the behavior of a heterogeneous virus system with additive i.i.d. Gaussian noise by using a homogeneous approximation, i.e. recovering the spread parameters by applying \eqref{eq:id1}. The second set of simulations illustrates some interesting behavior regarding the sampling parameter $h$.

\begin{figure}
    \centering
    \includegraphics[width=\linewidth]{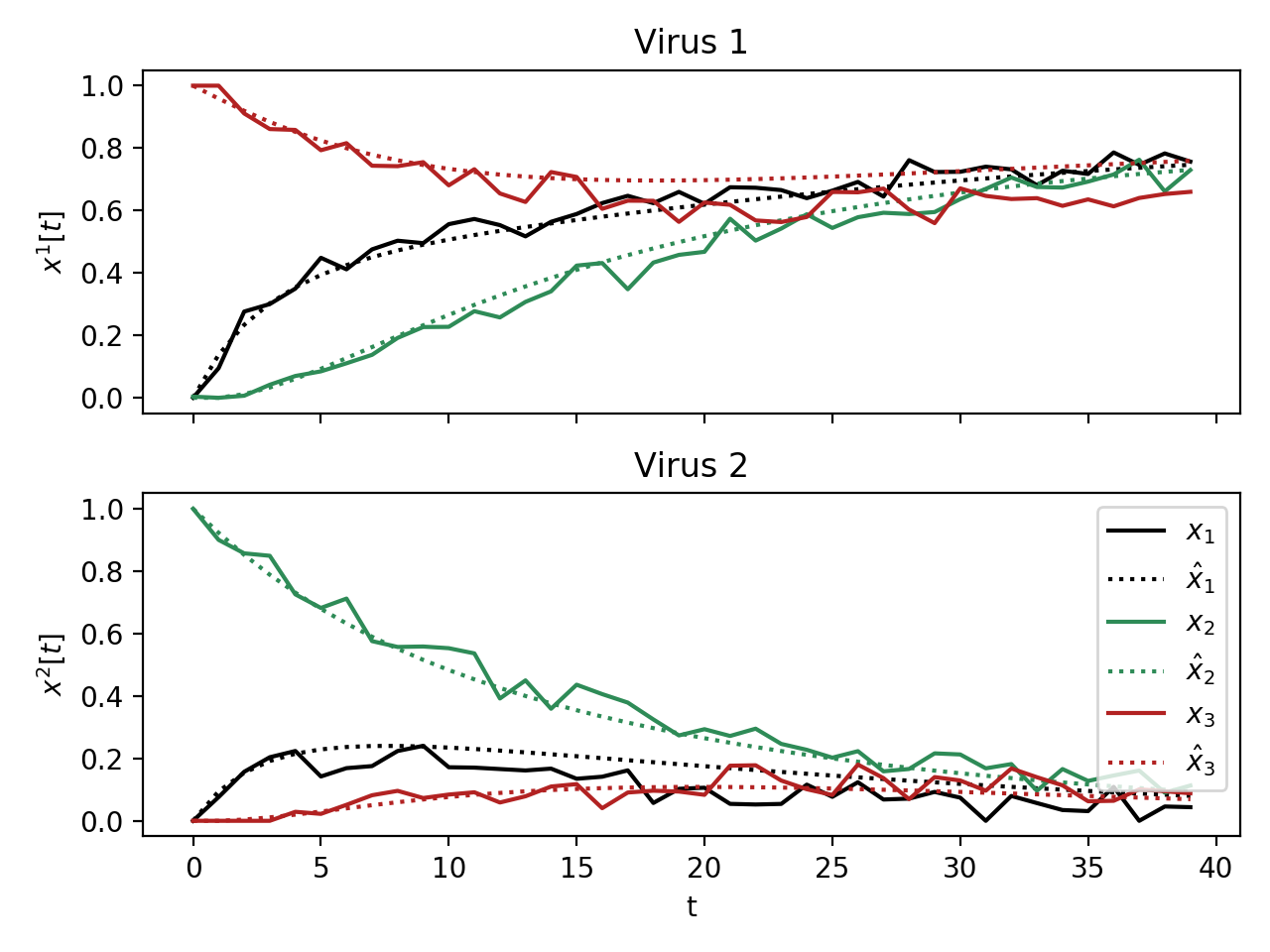}
    \caption{Simulation of the epidemic states of a heterogeneous system with additive i.i.d. Gaussian noise and recovered states using a homogeneous approximation of the system.}
    \label{fig:approx_noise}
\end{figure}

For the first simulation we consider a heterogeneous system with three agents and two viruses, $m = 2$. We set $h=1$,
\vspace{-2ex}
\begin{align*}
    x^1[0] &= \left[0 \; 0 \; 1 \right],
%
    \delta^1 = \left[ 0.05 \; 0.03 \; 0.04\right],
    \beta^1 = \left[ 0.15 \; 0.13 \; 0.08 \right],
    \\
    x^2[0] &= \left[ 0 \; 1 \; 0 \right],
    \delta^2 = \left[ 0.13 \; 0.07 \; 0.08 \right],
    \beta^2 = \left[ 0.09 \; 0.11 \; 0.10 \right]  \text{, and}
    \\
    A^1 &= A^2 = \begin{bmatrix} 0 & 1 & 1 \\ 1 & 0 & 1 \\ 1 & 1 & 0 \end{bmatrix}.
\end{align*}

\noindent We generate 40 time-steps of the epidemic states, $x$, using \eqref{eq:dis} with additive i.i.d Gaussian noise with the standard deviation set to $0.03$. 

To understand how accurately
a heterogeneous system can be approximated by a homogeneous model we use \eqref{eq:id1} with $T=4$
to learn homogeneous spread parameters. The learned parameters are 
\vspace{-1ex}
\begin{equation} \label{eq:learn_sim}
     \begin{bmatrix}  \hat{\delta}^1_h \\ \hat{\beta}^{1}_h \end{bmatrix} = \begin{bmatrix}      
    0.0415 \\ 0.1379 \end{bmatrix} \text{ and } \begin{bmatrix}  \hat{\delta}^{2}_h \\ \hat{\beta}^2_h \end{bmatrix} = \begin{bmatrix}  0.0772 \\  0.0944 \end{bmatrix}.
\end{equation}
\vspace{-1ex}

\noindent The learned parameters in \eqref{eq:learn_sim} are used to recover the generated data-samples, $\hat{x}$, by using \eqref{eq:dis} with homogeneous spread parameters. We compare $x$ and $\hat{x}$ in Figure~\ref{fig:approx_noise} to illustrate how well a homogeneous model can approximate a heterogeneous system with additive noise. We see that, even with noise in the system, the approximation is quite good. 

One can see that  the errors between the recovered states, $\hat{x}^2_1$ and $\hat{x}^1_3$, and the original system,  $x^2_1$ and $x^1_3$, are higher than the rest of the errors.
The decreased accuracy of $\hat{x}^1_3$ can be explained by the difference in magnitude of $\beta^1_3$ 
from the infection rates of the other agents. 
The same applies for $\hat{x}^2_1$ but for the healing rate, $\delta^2_1$.

\begin{figure}
    \centering
    \includegraphics[width=\linewidth]{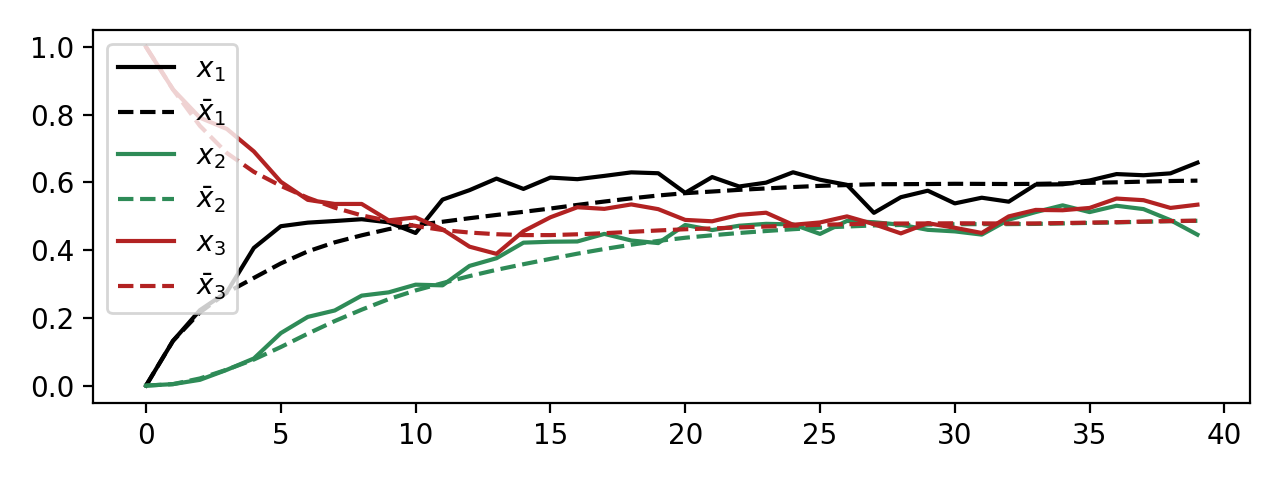}
    \caption{Simulation of the epidemic states of a homogeneous system with additive i.i.d. Gaussian noise and recovered states using the proposed online algorithm for learning the spread parameters, $\bar{x}$.}
    \label{fig:online_estimation}
\end{figure}

We now propose an online algorithm for learning the spread parameters from data, extending the ideas from Section \ref{sec:id}. The question becomes, if there is additive system noise and data is obtained in an online manner, how do estimates of the spread parameters improve as more data is added? The algorithm becomes the following: as more data is added, more rows of \eqref{eq:id1} are added. Then the spread parameters can be obtained by solving \eqref{eq:id1} at each time step, using least squares or by employing a recursive least squares method~\cite{aastrom2013adaptive}, to predict the next time step. 

A simulation based on the online algorithm for learning is shown in Figure~\ref{fig:online_estimation} with a single homogeneous virus. We set $h = 1$, $\delta = 0.9$, $\beta = 1.5$, $x[0] = \left[0 \; 0 \; 1 \right]$, and 

\vspace{-2ex}
\begin{align*}
    A &=\begin{bmatrix} 0 & 1 & 1 \\ 1 & 0 & 0 \\ 1 & 0 & 0 \end{bmatrix}.
\end{align*}

We generate the epidemic states, $x$, using (2) with additive i.i.d Gaussian noise with the standard deviation set to 0.03. By Figure~\ref{fig:online_estimation} we can see that the estimation is quite accurate using this online algorithm for learning, where $\bar{x}$ represents the estimated state. 
We can see that the new algorithm performs quite well, capturing the behavior of the system. We now apply these ideas to a USDA farm subsidy dataset.

\section{USDA Farm Subsidies as Competing Viruses}\label{sec:usda}

In the Food, Conservation and Energy Act of 2008 (2008 Farm Bill) a new subsidy program, ACRE, was introduced. It was an alternative to the exist CCP program. 
Similar to \cite{usda_acc,dtjournal} we aggregate farms on the county level. This approach allows us to convert the binary decision to enroll in ACRE or in CCP into a continuous measure of the proportion of eligible farms that enroll in ACRE or CCP, in each county. The proportion of farms enrolled in ACRE (and CCP) corresponds exactly to the density of the first virus (second virus), facilitating our investigation of the spread of the competing programs. The number of eligible farms in a county was set to the max number of farms enrolled in both programs in any year.  We removed counties where no farms were ever enrolled in either program. We also removed Alaska and Hawaii since they are not in the contiguous United States of America. The data for the four years considered can be found in Figures \ref{fig:data09}-\ref{fig:data12}.
Please see \cite{usda_acc,dtjournal} for more detailed information on the programs. 



\begin{figure}
  \begin{multicols}{2}
  \centering
  \begin{subfigure}{\linewidth}
    \includegraphics[trim=175 65 158 50, clip, width = .85\columnwidth]{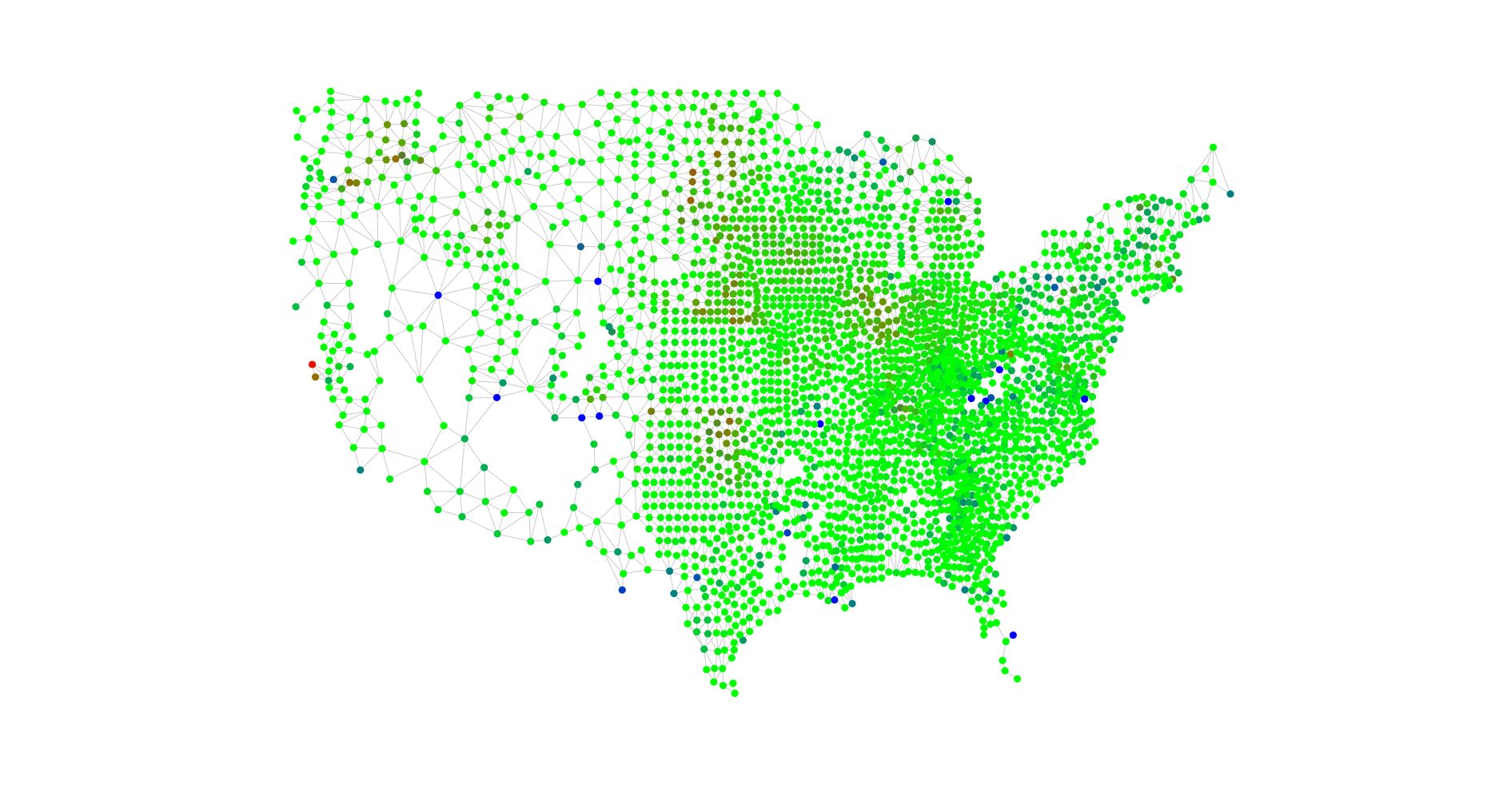}
    \caption{2009 Data}
    \label{fig:data09}
  \end{subfigure}
  \par
  \begin{subfigure}[b]{\columnwidth}
    \includegraphics[trim=175 65 158 50, clip, width = .85\columnwidth]{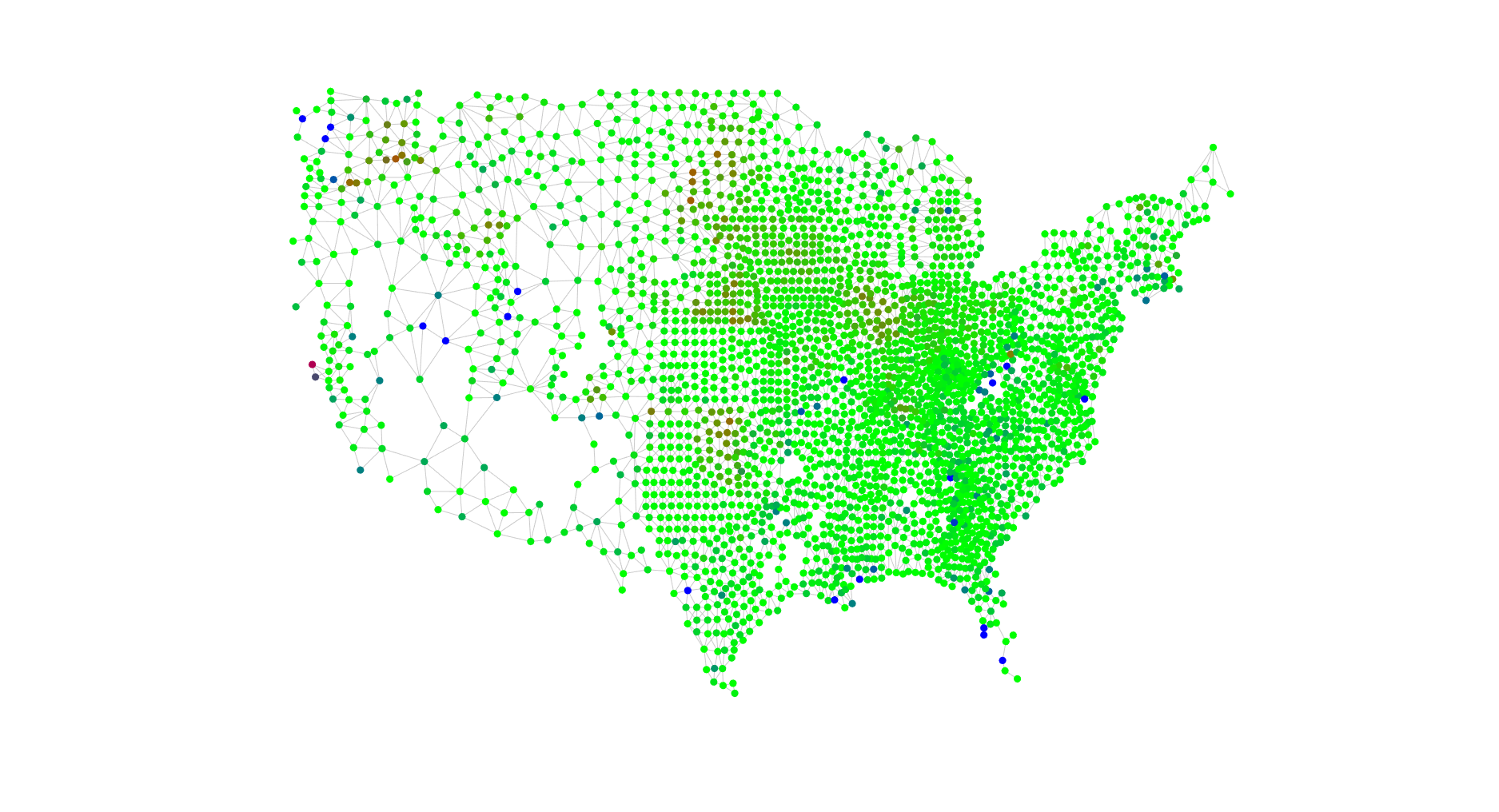}
    \caption{2010 Data}
    \label{fig:data10}
  \end{subfigure}
  \par
  \begin{subfigure}[b]{\columnwidth}
    \includegraphics[trim=175 65 158 50, clip, width = .85\columnwidth]{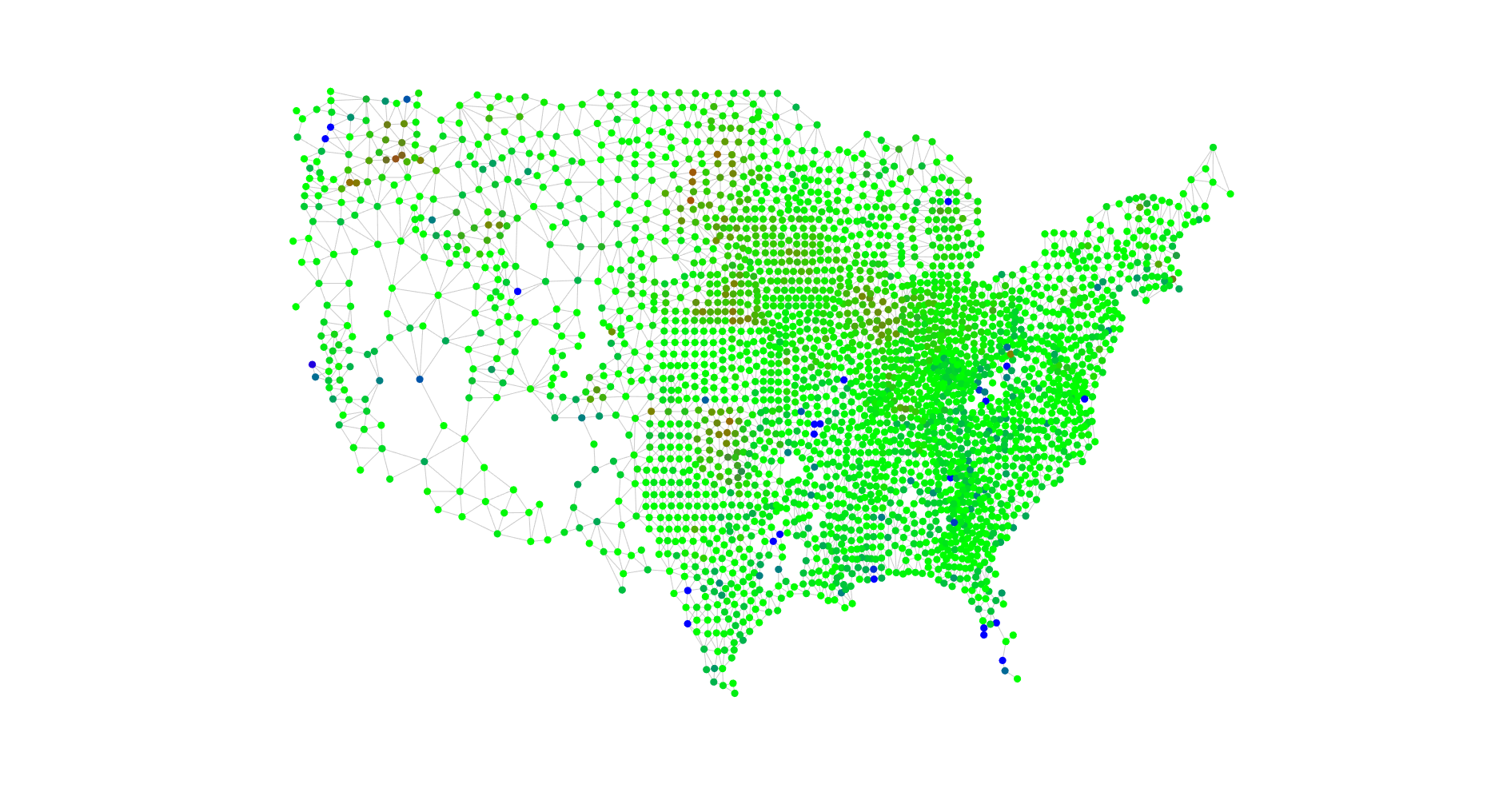}
    \caption{2011 Data}
    \label{fig:data11}
  \end{subfigure}
  \par
  \begin{subfigure}[b]{\columnwidth}
    \includegraphics[trim=175 65 158 50, clip, width = .85\columnwidth]{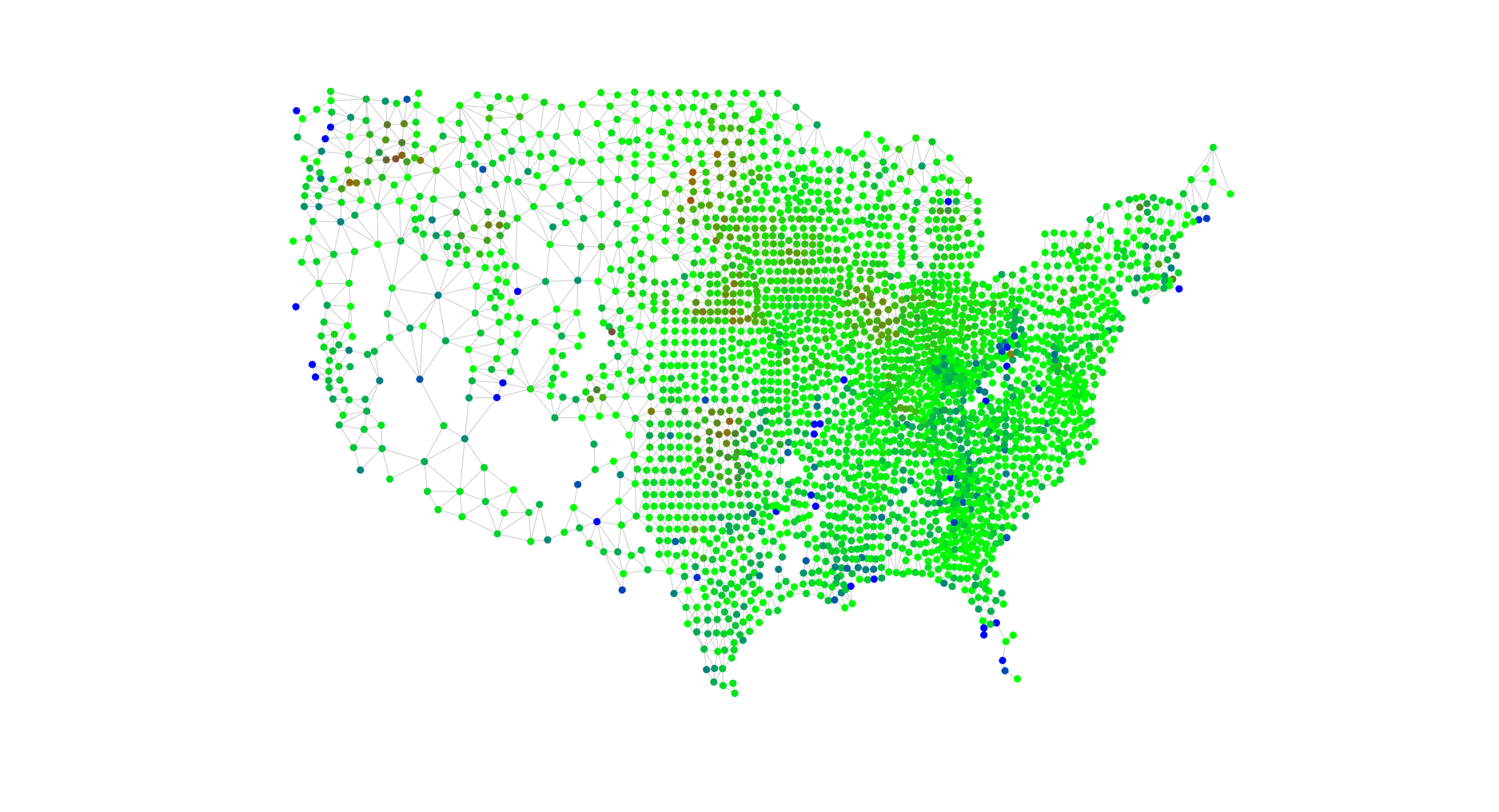}
    \caption{2012 Data}
    \label{fig:data12}
  \end{subfigure}
  \par
    \begin{subfigure}{\columnwidth}
      \includegraphics[trim=175 65 158 50, clip, width = .85\columnwidth]{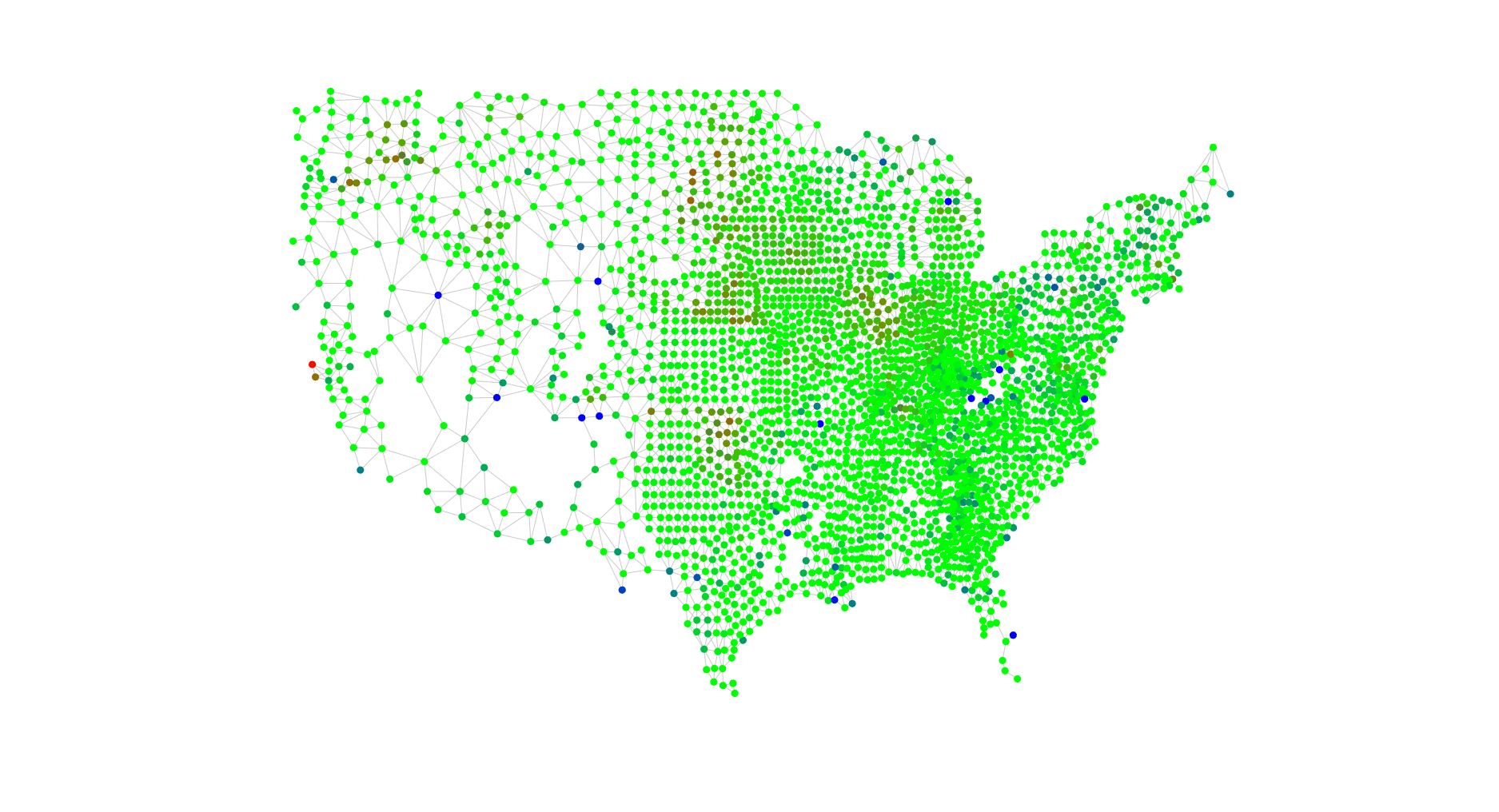}
      \caption{2009 Simulated Data}
      \label{fig:result09}
    \end{subfigure}
    \par
    \begin{subfigure}[b]{\columnwidth}
      \includegraphics[trim=175 65 158 50, clip, width = .85\columnwidth]{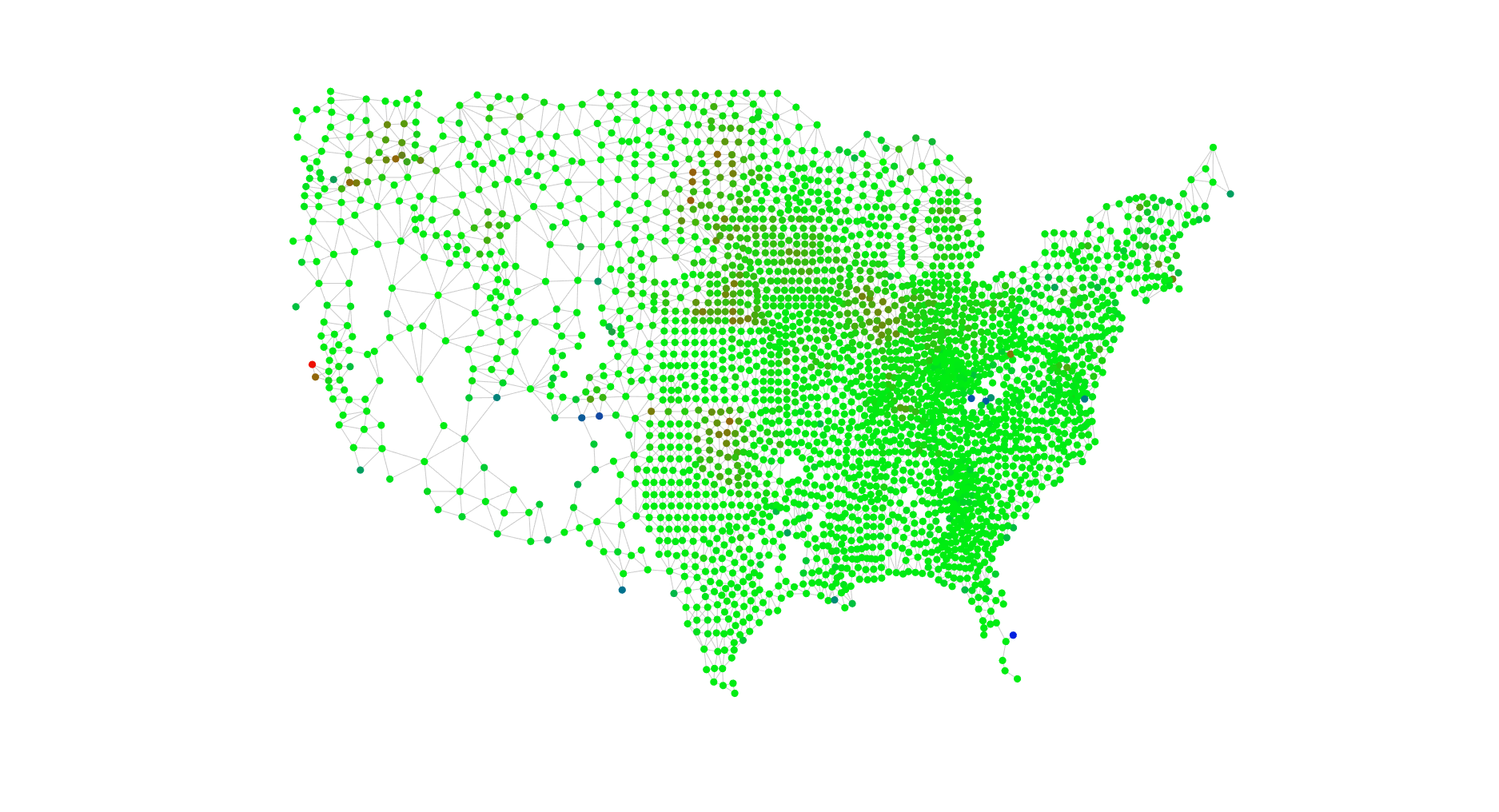}
      \caption{2010 Simulated Data}
      \label{fig:result10}
    \end{subfigure}
    \par
    \begin{subfigure}[b]{\columnwidth}
      \includegraphics[trim=175 65 158 50, clip, width = .85\columnwidth]{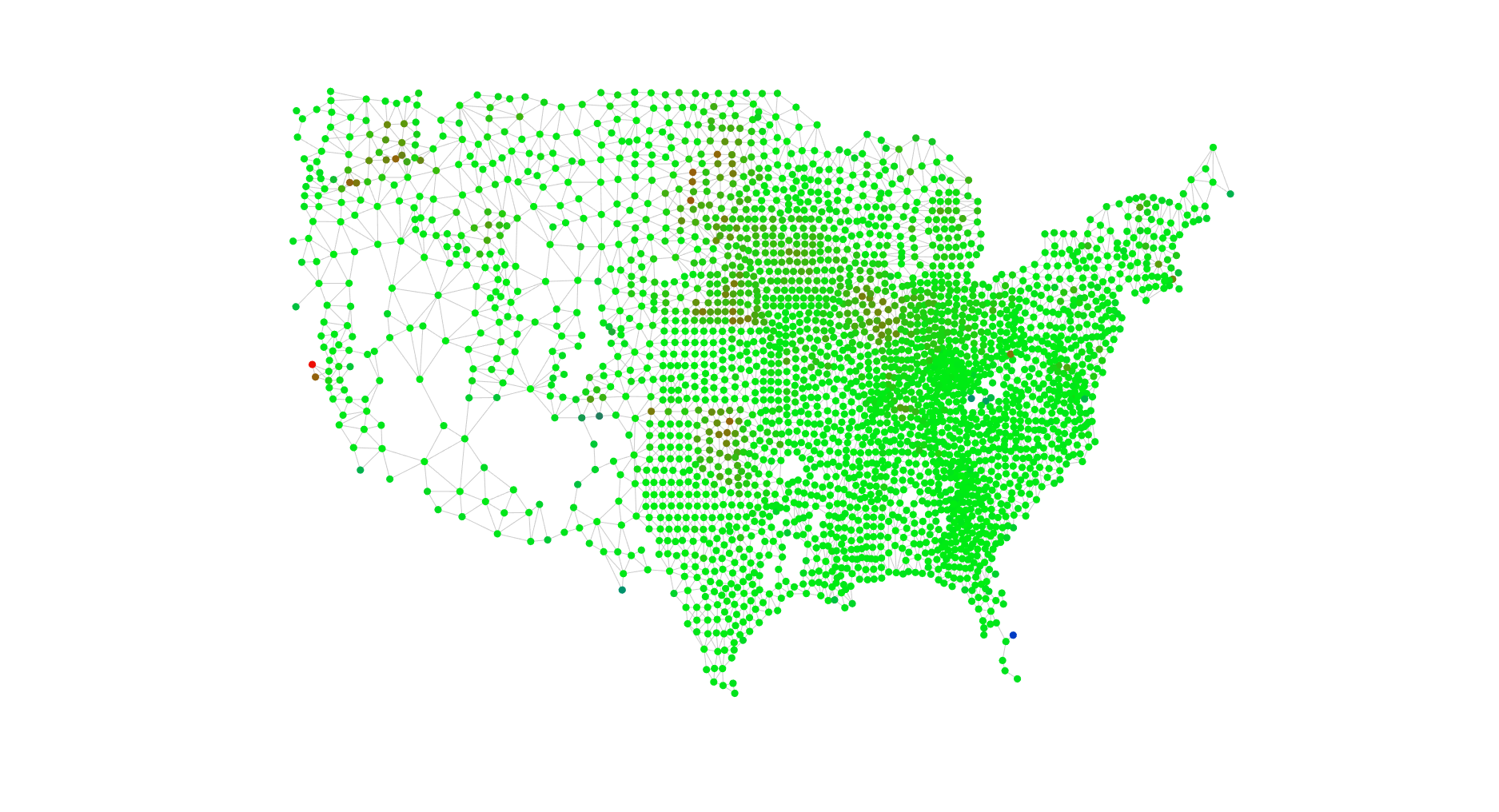}
      \caption{2011 Simulated Data}
      \label{fig:result11}
    \end{subfigure}
    \par
    \begin{subfigure}[b]{\columnwidth}
      \includegraphics[trim=175 65 158 50, clip, width = .85\columnwidth]{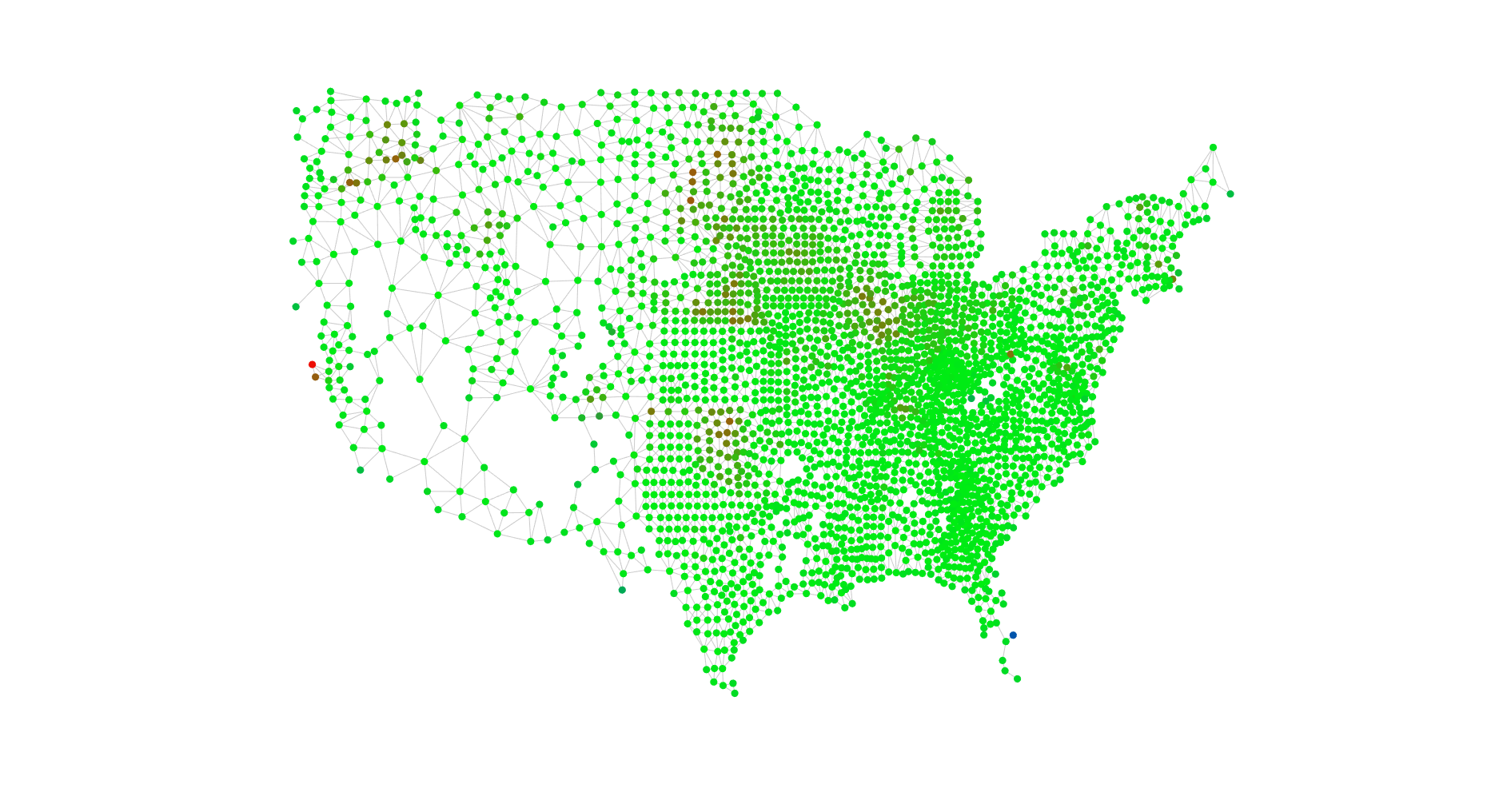}
      \caption{2012 Simulated Data}
      \label{fig:result12}
    \end{subfigure}
    \end{multicols}
  \caption{(Left) Calculated from the USDA dataset, the percentage of eligible farms enrolled in the ACRE Program, the CCP program, or neither are depicted in red, green, and blue, respectively. 
  (Right) Simulated data using Figure \ref{fig:data09} as the initial condition on the model in \eqref{eq:dis} with parameters calculated using the data from Kentucky, given in \eqref{eq:kentucky}. The colors of the nodes follow \eqref{eq:color}.}
  \label{fig:data}
\end{figure}

We now use the learning techniques presented in Section \ref{sec:id} and tested in Section \ref{sec:sim} for the model in \eqref{eq:dis} on the USDA dataset. 
The adjacency matrices are calculated using  the adjacency of counties, that is, 
\vspace{-1ex}
\begin{equation}
    a_{ij} = 
    \begin{cases}
    1, & \text{ if county }i \text{ and county } j \text{ share a border,}\\
    1, &\text{ if } i=j\text{,} \\
    0, & \text{ otherwise.}
    \end{cases}\label{eq:astate}
\end{equation} 


First we identify two sets of homogeneous spread parameters using the whole dataset by applying \eqref{eq:id1}:
\begin{equation}
     \begin{bmatrix}  \hat{\delta}^1_h \\ \hat{\beta}^{1}_h \end{bmatrix} = \begin{bmatrix}      
    0.0107 \\ 0.0139 \end{bmatrix} \text{ and } \begin{bmatrix}  \hat{\delta}^{2}_h \\ \hat{\beta}^2_h \end{bmatrix} = \begin{bmatrix}  0.0551 \\  0.0852 \end{bmatrix}.
    \label{eq:usa}
\end{equation}
We then simulate the model in \eqref{eq:dis} with the spread parameters in \eqref{eq:usa}, with the data from Figure \ref{fig:data09} being used as the initial condition. 
The resulting  scaled error between the dataset, $\mathbb{F}$, and the simulated data, $\hat{\mathbb{F}}_{\text{all}}$, using the Frobenius norm is 
\begin{equation*}
    \frac{\left\| \mathbb{F} - \hat{\mathbb{F}}_{\text{all}} \right\|_{F}}{\left\| \mathbb{F} \right\|_{F}} = \frac{12.0420}{96.8382} = 0.1244.
\end{equation*}

For completeness, similar to \cite{usda_acc,dtjournal} we use a subset of the dataset, the USDA data from Idaho, 
to recover the two sets of homogeneous model parameters and then simulate the spread of programs over the whole contiguous United States using the learned parameters. 
For calculating the adjacency matrix for Idaho, adjacent counties from bordering states were ignored. 
Applying \eqref{eq:id1} on the Idaho dataset gives the following spread parameters:
\vspace{-1ex}
\begin{equation}
     \begin{bmatrix}  \hat{\delta}^1_h \\ \hat{\beta}^1_h \end{bmatrix} = \begin{bmatrix}  -0.0332 \\ \ 0.0663 \end{bmatrix} \text{ and } \begin{bmatrix}  \hat{\delta}^2_h \\ \hat{\beta}^2_h \end{bmatrix} = \begin{bmatrix}  0.0503 \\   0.0345 \end{bmatrix}.
    \label{eq:idaho}
\end{equation}
Note that $\hat{\delta}^1_h$ for the first virus (the ACRE program) is negative, violating the assumptions of the model, which is not ideal.
Nevertheless for completeness, we simulate the spread over the contiguous United States using the model in \eqref{eq:dis} with the spread parameters calculated using the data from Idaho, given in \eqref{eq:idaho}, with the data from Figure \ref{fig:data09} being used as the initial condition. 
The scaled error between the dataset, $\mathbb{F}$, and this simulated data, $\hat{\mathbb{F}}_{\text{ID}}$, is 
\begin{equation*}
    \frac{\left\| \mathbb{F} - \hat{\mathbb{F}}_{\text{ID}} \right\|_{F}}{\left\| \mathbb{F} \right\|_{F}} = \frac{14.28}{96.8382} = 0.1348.
\end{equation*}
The scaled error from the analogous simulation in \cite{usda_acc,dtjournal} was 0.2348. Therefore it would appear that, while not a perfect fit, the competitive-virus model seems to capture the behavior of this USDA Farm Subsidy adoption dataset better than the single virus model. 

After testing every possible state, we found that the data from Kentucky provided the best estimate of the whole US data set when using the homogeneous version of the model in \eqref{eq:disM}. 
Applying \eqref{eq:id1} on the Kentucky dataset gives the following spread parameters:
\vspace{-1ex}
\begin{equation}
     \begin{bmatrix}  \hat{\delta}^1_h \\ \hat{\beta}^1_h \end{bmatrix} = \begin{bmatrix}  0.0044 \\ 0.1352 \end{bmatrix} \text{ and } \begin{bmatrix}  \hat{\delta}^2_h \\ \hat{\beta}^2_h \end{bmatrix} = \begin{bmatrix}  0.0702 \\   0.1272 \end{bmatrix}.
    \label{eq:kentucky}
\end{equation}
The simulated data can be found in Figures \ref{fig:result09}-\ref{fig:result12}. The resulting scaled error between the dataset, $\mathbb{F}$, and the simulated data, $\hat{\mathbb{F}}_{\text{KY}}$, is 
\begin{equation*}
    \frac{\left\| \mathbb{F} - \hat{\mathbb{F}}_{\text{KY}} \right\|_{F}}{\left\| \mathbb{F} \right\|_{F}} = \frac{12.2724}{96.8382} = 0.1230.
\end{equation*}


The results were improved upon when implementing the recursive algorithm proposed in Section~\ref{sec:exp}, reducing the scaled error to $0.0855$. However, it must be noted that the first two data points were included in the simulated data, 
since the recursive algorithm is only used for one step prediction. Using the first set of learned spread parameters for the second and third data points gave an error of 0.1140, still improving upon the previous results.

\section{Conclusion}\label{sec:con}

\vspace{-.5ex}

In this work we have proposed a discrete time competing virus model for an arbitrary number of viruses. We have provided conditions for the model to be well defined. 
We provided necessary and sufficient conditions for uniqueness of the healthy equilibrium. 
We presented necessary and sufficient conditions for learning spread parameters for competing viruses from data. We presented an interesting set of simulations that illustrate the analytic results and depict some characteristics of the model that warrant further study, 
and proposed an online spread parameter estimation algorithm. 
We employed a previously studied USDA dataset to validate the discrete-time two-competing virus, or bi-virus, case by modeling the spread of two alternative farm subsidy programs among farms aggregated by county, improving on previous work. 


\vspace{-2ex}

\bibliographystyle{IEEEtran}
\bibliography{IEEEabrv,bib}

\begin{thebibliography}{10}
\providecommand{\url}[1]{#1}
\csname url@samestyle\endcsname
\providecommand{\newblock}{\relax}
\providecommand{\bibinfo}[2]{#2}
\providecommand{\BIBentrySTDinterwordspacing}{\spaceskip=0pt\relax}
\providecommand{\BIBentryALTinterwordstretchfactor}{4}
\providecommand{\BIBentryALTinterwordspacing}{\spaceskip=\fontdimen2\font plus
\BIBentryALTinterwordstretchfactor\fontdimen3\font minus
  \fontdimen4\font\relax}
\providecommand{\BIBforeignlanguage}[2]{{%
\expandafter\ifx\csname l@#1\endcsname\relax
\typeout{** WARNING: IEEEtran.bst: No hyphenation pattern has been}%
\typeout{** loaded for the language `#1'. Using the pattern for}%
\typeout{** the default language instead.}%
\else
\language=\csname l@#1\endcsname
\fi
#2}}
\providecommand{\BIBdecl}{\relax}
\BIBdecl

\bibitem{allcott2017social}
H.~Allcott and M.~Gentzkow, ``Social media and fake news in the 2016
  election,'' \emph{Journal of Economic Perspectives}, vol.~31, no.~2, pp.
  211--36, 2017.

\bibitem{bastos2019brexit}
M.~T. Bastos and D.~Mercea, ``The {Brexit} botnet and user-generated
  hyperpartisan news,'' \emph{Social Science Computer Review}, vol.~37, no.~1,
  pp. 38--54, 2019.

\bibitem{nowak1991evolution}
M.~Nowak, ``The evolution of viruses. competition between horizontal and
  vertical transmission of mobile genes,'' \emph{Journal of Theoretical
  Biology}, vol. 150, no.~3, pp. 339--347, 1991.

\bibitem{sahneh2014competitive}
F.~D. Sahneh and C.~Scoglio, ``Competitive epidemic spreading over arbitrary
  multilayer networks,'' \emph{Physical Review E}, vol.~89, no.~6, p. 062817,
  2014.

\bibitem{karrer2011competing}
B.~Karrer and M.~E.~J. Newman, ``Competing epidemics on complex networks,''
  \emph{Physical Review E}, vol.~84, no.~3, p. 036106, 2011.

\bibitem{prakash2012winner}
B.~A. Prakash, A.~Beutel, R.~Rosenfeld, and C.~Faloutsos, ``Winner takes all:
  competing viruses or ideas on fair-play networks,'' in \emph{Proceedings of
  the 21st International Conference on World Wide Web}.\hskip 1em plus 0.5em
  minus 0.4em\relax ACM, 2012.

\bibitem{wei2013competing}
X.~Wei, N.~C. Valler, B.~A. Prakash, I.~Neamtiu, M.~Faloutsos, and
  C.~Faloutsos, ``Competing memes propagation on networks: A network science
  perspective,'' \emph{IEEE Journal on Selected Areas in Communications},
  vol.~31, no.~6, pp. 1049--1060, 2013.

\bibitem{santos2015bivirus}
A.~Santos, J.~Moura, and J.~Xavier, ``Bi-virus {SIS} epidemics over networks:
  Qualitative analysis,'' \emph{IEEE Transactions on Network Science and
  Engineering}, vol.~2, no.~1, pp. 17--29, Jan 2015.

\bibitem{liu2016onthe}
J.~Liu, P.~E. Par\'{e}, , A.~Nedi\'{c}, C.~Y. Tang, C.~L. Beck, and
  T.~Ba\c{s}ar, ``On the analysis of a continuous-time bi-virus model,'' in
  \emph{Proceedings of the 55th IEEE Conference on Decision and Control (CDC)},
  2016.

\bibitem{watkins2016optimal}
N.~J. Watkins, C.~Nowzari, V.~M. Preciado, and G.~J. Pappas, ``Optimal resource
  allocation for competitive spreading processes on bilayer networks,''
  \emph{IEEE Transactions on Control of Network Systems}, vol.~5, no.~1, pp.
  298--307, 2016.

\bibitem{bivirusTAC}
J.~Liu, P.~E. Par\'{e}, A.~Nedi\'{c}, C.~T. Tang, C.~L. Beck, and T.~Ba\c{s}ar,
  ``Analysis and control of a continuous-time bi-virus model,'' \emph{IEEE
  Transactions on Automatic Control}, 2019, to appear.

\bibitem{xu2012multi}
S.~Xu, W.~Lu, and Z.~Zhan, ``A stochastic model of multivirus dynamics,''
  \emph{IEEE Transactions on Dependable and Secure Computing}, vol.~9, no.~1,
  pp. 30--45, 2012.

\bibitem{acc_multi}
P.~E. Par\'{e}, J.~Liu, C.~L. Beck, A.~Nedi\'{c}, and T.~Ba\c{s}ar,
  ``Multi-competitive viruses over static and time--varying networks,'' in
  \emph{Proceedings of the American Control Conference}, 2017, pp. 1685--1690.

\bibitem{tnse_multi}
------, ``Multi-competitive viruses over time--varying networks with mutations
  and human awareness,'' \emph{{\em under review for} IFAC Automatica}, 2019.

\bibitem{wang2003epidemic}
Y.~Wang, D.~Chakrabarti, C.~Wang, and C.~Faloutsos, ``Epidemic spreading in
  real networks: An eigenvalue viewpoint,'' in \emph{Proceedings of the 22nd
  International Symposium on Reliable Distributed Systems}.\hskip 1em plus
  0.5em minus 0.4em\relax IEEE, 2003, pp. 25--34.

\bibitem{chakrabarti2008epidemic}
D.~Chakrabarti, Y.~Wang, C.~Wang, J.~Leskovec, and C.~Faloutsos, ``Epidemic
  thresholds in real networks,'' \emph{ACM Transactions on Information and
  System Security (TISSEC)}, vol.~10, no.~4, p.~1, 2008.

\bibitem{ahn2013global}
H.~J. Ahn and B.~Hassibi, ``Global dynamics of epidemic spread over complex
  networks,'' in \emph{Proceedings of the 52nd IEEE Conference on Decision and
  Control (CDC)}.\hskip 1em plus 0.5em minus 0.4em\relax IEEE, 2013, pp.
  4579--4585.

\bibitem{usda_acc}
P.~E. Par\'{e}, B.~E. Kirwan, J.~Liu, T.~Ba\c{s}ar, and C.~L. Beck,
  ``Discrete-time spread processes: Analysis, identification, and validation,''
  in \emph{Proceedings of the American Control Conference}, 2018.

\bibitem{dtjournal}
P.~E. Par\'{e}, J.~Liu, C.~L. Beck, B.~E. Kirwan, , and T.~Ba\c{s}ar,
  ``Discrete time virus spread processes: Analysis, identification, and
  validation,'' \emph{IEEE Transactions on Control Systems Technology: System
  Identification and Control in Biomedical Applications}, vol.~28, no.~1, pp.
  79--93, 2019.

\bibitem{arturo}
A.~Melo, C.~L. Beck, J.~I. Pena, and P.~E. Par\'{e}, ``Knowledge transfer from
  universities to regions as a network spreading process,'' in
  \emph{Proceedings of the IEEE International Systems Engineering Symposium
  (ISSE)}, 2018.

\bibitem{prasse2018network}
B.~Prasse and P.~Van~Mieghem, ``Network reconstruction and prediction of
  epidemic outbreaks for {NIMFA} processes,'' \emph{arXiv preprint
  arXiv:1811.06741}, 2018.

\bibitem{prasse2019viral}
------, ``The viral state dynamics of the discrete-time {NIMFA} epidemic
  model,'' \emph{arXiv preprint arXiv:1903.08027}, 2019.

\bibitem{atkinson2008introduction}
K.~E. Atkinson, \emph{An introduction to numerical analysis}.\hskip 1em plus
  0.5em minus 0.4em\relax John Wiley \& Sons, 2008.

\bibitem{rantzer2011positive}
A.~Rantzer, ``Distributed control of positive systems,'' in \emph{Proceedings
  of the 50th IEEE Conference on Decision and Control}, 2011, pp. 6608--6611.

\bibitem{horn2012matrix}
R.~A. Horn and C.~R. Johnson, \emph{Matrix analysis}.\hskip 1em plus 0.5em
  minus 0.4em\relax Cambridge University Press, 2012.

\bibitem{aastrom2013adaptive}
K.~J. {\AA}str{\"o}m and B.~Wittenmark, \emph{Adaptive control}.\hskip 1em plus
  0.5em minus 0.4em\relax Courier Corporation, 2013.

\end{thebibliography}

\end{document}